\documentclass[runningheads]{llncs}

\usepackage{syntax}
\usepackage{amsmath}
\usepackage{amssymb}
\usepackage{mathtools}
\usepackage{bussproofs}
\usepackage[table]{xcolor}
\usepackage{todonotes}
\usepackage{listings}
\usepackage{tikz,tikzsymbols}
\usepackage[inline]{enumitem}
\usepackage{abs}
\usepackage{jmllistings}
\usepackage{soul}
\usepackage[hidelinks]{hyperref}
\usepackage{tikz}
\usetikzlibrary{tikzmark}
\lstdefinelanguage{Async}{
basicstyle=\small\normalfont\ttfamily,
numbers=left,
frame=lines,
mathescape=true,
numberstyle=\tiny,
comment=[l]{//},
commentstyle={\color{gray}\ttfamily},
identifierstyle=\color{black},
keywords={if, while, return, open,close,read,write,skip},
keywordstyle={\color{blue}\bfseries},
morecomment=[s]{/*}{*/},
morestring=[b]",
morestring=[s]{"""*}{*"""},
sensitive=true,
stringstyle={\color{brown}\ttfamily},
}

\newcommand{\async}[1]{\lstinline[language=Async]{#1}}
\newcommand{\xasync}[1]{\text{\lstinline[language=Async]{#1}}}

\newcommand{\LVar}{\ensuremath{\mathsf{LVar}}}

\newcommand{\RVar}{\ensuremath{\mathsf{RecVar}}}
\newcommand{\stateFml}[1]{\lceil#1\rceil}
\newcommand{\chop}{\ensuremath*\!*}
\newcommand{\concat}{\ensuremath\cdot}

\newcommand{\muFml}[2]{\ensuremath{\mu#1.#2}}

\newcommand{\interp}{\ensuremath{\mathsf{I}}}

\newcommand{\upP}[1]{\upl#1\upr}
\newcommand{\eval}[2]{\ensuremath{[\![#1]\!]_{#2}}}

\newcommand{\locEval}[2]{\ensuremath{[\![#1]\!]_{#2}}^{L}}
\newcommand{\globEval}[2]{\ensuremath{[\![#1]\!]_{#2}}^{G}}

\newcommand{\evalPhi}[4]{\ensuremath{[\![#1]\!]^{#2}_{#3,#4}}}
\newcommand{\evalPhiS}[1]{\ensuremath{\evalPhi{#1}{
  }{\rho}{o}}}

\newcommand{\evalPhiNoArgBig}[1]{\ensuremath{\left[\!\!\left[#1\right]\!\!\right]}}
\newcommand{\evalPhiNoArg}[1]{\ensuremath{[\![#1]\!]}}

\newcommand{\singleton}[1]{\ensuremath{\langle#1\rangle}}
\newcommand{\chopT}{\ensuremath{\underline{\chop}}}
\newcommand{\States}{\ensuremath{\Sigma}}
\newcommand{\Traces}{\ensuremath{\mathsf{Traces}}}

\newcommand{\noEvent}[1]{\stackrel{#1}{\ensuremath{\cdot\cdot}}}

\newcommand{\Todo}[2][]{%
\ifthenelse{\equal{#1}{done}}
  {
    \todo[backgroundcolor=green!30,inline]{#2}
  }
    {
      \ifthenelse{\equal{#1}{warning}}
      {
        \todo[backgroundcolor=orange!30,inline]{#2}
      }
      {
        {
          \ifthenelse{\equal{#1}{ongoing}}
          {
            \todo[backgroundcolor=blue!30,inline]{#2}
          }
          {
            \todo[backgroundcolor=red!30,inline]{#2}
          }
        }
      }
    }
}

\usepackage[nonumberofruns]{xassoccnt}

\NewTotalDocumentCounter{TodosCnt}
\NewTotalDocumentCounter{TodosDone}

\newcommand{\pushEv}{\ensuremath{\mathsf{push}}\xspace}
\newcommand{\popEv}{\ensuremath{\mathsf{pop}}\xspace}
\newcommand{\callEv}{\ensuremath{\mathsf{call}}\xspace}
\newcommand{\retEv}{\ensuremath{\mathsf{ret}}\xspace}
\newcommand{\startEv}{\ensuremath{\mathsf{start}}\xspace}

\newcommand{\invocEv}{\ensuremath{\mathsf{invoc}}\xspace}

\newcommand{\upl}{\ensuremath{\{}}
\newcommand{\upr}{\ensuremath{\}}}
\newcommand{\update}{\ensuremath{\mathcal{U}}}

\newcommand{\judge}[2]{\ensuremath{#1:#2}}

\newcommand{\finiteNoEv}[1]{\ensuremath{\overset{\begin{subarray}{l}#1\end{subarray}}}{\cdot\cdot\cdot}}

\newcommand{\turnstyle}{\vdash}

\newcommand{\ruleName}[1]{{\ensuremath{\mathsf{#1}}}}

\newcommand{\seq}[2]{#1\turnstyle#2}
\newcommand{\semseq}[2]{#1\models#2}

\newcommand{\sequent}[2]{\seq{\Gamma\ifx#1\relax\else,#1\fi} 
                             {#2}}
\newcommand{\semsequent}[2]{\semseq{\Gamma\ifx#1\relax\else,#1\fi} 
                                   {#2}}

\newcommand{\adfrac}[2]{\genfrac{}{}{}{0}
  {\begin{array}{l}#1\end{array}}
  {\begin{array}{l}#2\end{array}}}

\newcommand{\seqRule}[3]{\mbox{\small\ruleName{#1}}\ 
  \adfrac{#2}{#3}}

\newcommand{\Simple}[1]{\mbox{\lstinline{#1}}}
\newcommand{\code}[1]{\mbox{\lstinline|#1|}}

\newcommand{\mycomment}[1]{}

\newcommand{\methods}[1]{\method(#1)}
\newcommand{\Gi}{\ensuremath \mathcal{G}}

\newcommand{\last}{\ensuremath \mathrm{last}}

\newcommand{\cont}[1]{\ensuremath \mathrm{K}(#1)}

\newcommand{\contP}[2]{\ensuremath \mathrm{K}^{#1}(#2)}
\newcommand{\contF}[2]{\ensuremath \mathrm{K}^{#1}(#2)}
\newcommand{\zero}{\ensuremath \circ}
\renewcommand{\zero}{\mbox{\bottle}}

\newcommand{\trueSem}{\ensuremath \mathrm{t\!t}}
\newcommand{\falseSem}{\ensuremath \mathrm{f\!f}}

\newcommand{\valB}[2]{\ensuremath \mathrm{val}_{#1}(#2)}

\newcommand{\valP}[2]{\ensuremath \mathrm{val}_{\sigma}^{#1}(#2)}
\newcommand{\valPC}[3]{\ensuremath \mathrm{val}_{\sigma}^{#1,#2}(#3)}
\newcommand{\valBP}[3]{\ensuremath \mathrm{val}_{#1}^{#2}(#3)}

\newcommand{\valDnoargs}{\ensuremath \mathrm{val}_{\sigma}}
\newcommand{\chopSem}{\underline{\mathbin{\ast\ast}}}
\newcommand{\updatestate}[3]{#1[#2 \mapsto #3]}
\newcommand{\chopTrSem}[2]{\ensuremath #1 \chopSem #2}
\newcommand{\concatTr}[2]{\ensuremath #1\cdot #2}
\newcommand{\consTr}[2]{\ensuremath #1 \curvearrowright #2}
\newcommand{\cons}{\ensuremath \curvearrowright}

\newcommand{\UNUSED}[1]{}
\newcommand{\OUTDATED}[1]{}

\DeclareMathOperator{\method}{procedures}

\DeclareMathOperator{\lp}{lookup}

\newcommand{\lookup}[1]{\lp(#1)}

\newcommand{\ifStmt}[2]{\Simple{if} \ #1 \ \{~ #2 ~\}\xspace}

\newcommand{\assignStmt}[2]{#1=#2\xspace}

\newcommand{\returnStmt}[1]{\xasync{return}\ #1\xspace}
\newcommand{\openStmt}[1]{\Simple{open}(#1)\xspace}
\newcommand{\closeStmt}[1]{\Simple{close}(#1)\xspace}
\newcommand{\readStmt}[1]{\Simple{read}(#1)\xspace}
\newcommand{\writeStmt}[1]{\Simple{write}(#1)\xspace}

\newcommand{\semPair}[2]{#1,\,K(#2)}

\newcommand{\semSeq}[3]{#1 \overset{#3}{\to} #2}

\newcommand{\inline}{\mathrm{inline}}

\newcommand{\event}[2]{\mathsf{#1}(#2)}
\newcommand{\invocEvP}[2]{\mathsf{invoc}_{#1}(#2)} 
\newcommand{\callEvP}[2]{\mathsf{call}_{#1}(#2)}
\newcommand{\runEv}{\mathsf{run}}
\newcommand{\pushEvP}[2]{\mathsf{push}_{#1}(#2)}
\newcommand{\popEvP}[2]{\mathsf{pop}_{#1}(#2)}

\newcommand{\openEvP}[2]{\mathsf{open}_{#1}(#2)} 
\newcommand{\closeEvP}[2]{\mathsf{close}_{#1}(#2)}
\newcommand{\readEvP}[2]{\mathsf{read}_{#1}(#2)} 
\newcommand{\writeEvP}[2]{\mathsf{write}_{#1}(#2)}

\newcommand{\evTrioPV}[3]{\textit{ev}_{#1}^{#3}({#2})} 

\newcommand{\ev}[1]{\textit{ev}(#1)} 
\newcommand{\evp}[2]{\textit{ev}^{#1}(#2)} 
\newcommand{\stateorevent}{t}

\newcommand{\num}[2]{\ensuremath \#_{#1}({#2})}
\newcommand{\many}[1]{\overline{#1}}

\newcommand{\gcondrule}[4]{ 
  \begin{equation}
    \label{eq:rule#1}
    \textsc{({#2})}\,
  \begin{array}{c} 
    #3 \\[1pt] 
    \hline\\[-7pt]
    #4 
  \end{array} 
 \end{equation}
    }

\newcommand*{\opf}{\event{op}{f}}

\newcommand*{\clf}{\event{cl}{f}}
\newcommand*{\workf}{\event{w}{f}}

\usepackage{calc}
\usepackage{xspace}

\usepackage{todonotes}

\usepackage{wrapfig}

\newcommand{\assumeAndContinue}[3]{\ensuremath{\assume{#1} \ | \ #2 \ | \ \continue{#3}}}

\newcommand{\assumeAndContinueShort}[3]{\ensuremath{\assume{#1} \big|\big. #2 \big|\big. \continue{#3}}}
\newcommand{\assume}[1]{\ensuremath{<\hspace{-5pt}<\!#1}}

\newcommand{\continue}[1]{\ensuremath{#1\!>\hspace{-5pt}>}}

\newcommand{\GlobalJudge}[2]{\ensuremath{#1:_G#2}}

\title{Context-aware Trace Contracts}

\author{Reiner H\"ahnle\inst{1} \and Eduard Kamburjan\inst{2} \and Marco Scaletta\inst{1}}

\institute{Technical University of Darmstadt, Germany,
 \email{<firstName>.<lastName>@tu-darmstadt.de} \and University of Oslo, Norway, \email{eduard@ifi.uio.no}}

\begin{document}

\maketitle


\begin{abstract}
  The behavior of concurrent, asynchronous procedures depends in
  general on the call context, because of the global protocol that
  governs scheduling. This context cannot be specified with the
  state-based Hoare-style contracts common in deductive verification.
  Recent work generalized state-based to trace contracts, which permit
  to specify the internal behavior of a procedure, such as calls or state
  changes, but not its call context.
  In this article we propose a program logic of context-aware trace
  contracts for specifying global behavior of asynchronous
  programs.
  We also provide a sound proof system that addresses two challenges:
  To observe the program state not merely at the end points of a
  procedure, we introduce the novel concept of an observation
  quantifier. And to combat combinatorial explosion of possible call
  sequences of procedures, we transfer Liskov's principle of
  behavioral subtyping to the analysis of asynchronous procedures.
  
\end{abstract}
%

\section{Introduction}\label{sec:introduction}


Contracts \cite{Larch93,Meyer92} are a cornerstone of the
rely-guarantee paradigm for verification \cite{Jones81}, as it enables
the decomposition of a program along naturally defined boundaries.  A
well-established example are procedure contracts which encapsulate the
behavior of the execution of a single procedure in terms of pre- and
post-condition.\footnote{Additional specification elements, such as
  frames or exceptional behavior, can be considered as syntactic sugar
  to achieve concise post-conditions.}  Traditionally, a pre-condition
describes the state at the moment a procedure is called, and the
post-condition describes the state at the moment the procedure
terminates. The contract-based approach to deductive verification
permits to verify a program in a \emph{procedure-modular} manner and
so makes it possible to conduct correctness proofs of sequential
programs of considerable complexity and size in real programming
languages \cite{TimSort17,HaehnleHuisman19}.

A concurrent setting poses a completely different challenge, because
concurrently executing procedures may interfere on the state, causing
a myriad of possibly different behaviors, even for small
programs. Contracts in the presence of fine-grained concurrency tend
to be highly complex, because they have to encode substantial parts of
the invariants of the whole system under
verification~\cite{BBBB12}. It was rightly argued since long that a
suitable \emph{granularity} of interference is key to arrive at
manageable specifications of concurrent programs~\cite{Jones95b}.

\paragraph{Main Contribution.}
  
The present paper constitutes an effort to generalize the
contract-based approach to the verification for concurrent
programs. Generalization involves two aspects: First, it is necessary
to specify sets of \emph{traces} of a verified program, not merely
sets of pre-/post states, to be able to refer to the context of a
program and to internal events, such as other procedure
calls. Trace-based logics, such as temporal logic, are standard to
specify concurrent programs. Here, we use a recent trace logic
\cite{DBLP:journals/corr/abs-2211-09487} that is exogenous
\cite{Pnueli77} (i.e., allows judgments involving explicit programs)
and can characterize procedure calls.
Second, and this is our central contribution, not merely the procedure
under verification is specified by traces, but
\emph{pre-/postconditions are generalized to traces}. This permits to
specify the context in which a contract is supposed to be applied
without ghost variables or other auxiliary constructs. As we are going
to show, this approach admits a \emph{procedure-modular} deductive
verification system for concurrent programs, where the correctness of
each procedure contract implies the correctness of the conjunction of
all contracts, i.e.\ of the whole program.

\paragraph{Setting.}

There is an important limitation: our current approach does not work
for preemptive concurrency, but is targeted at the \emph{active
  object} paradigm \cite{ActiveObjects17}. With asynchronous
procedures and syntactically explicit suspension points, active
objects are a good trade-off between usability and verifiability, and
they feature appropriate \emph{granularity} to render contract-based
specification useful.

Earlier research \cite{DinBH15,DinOwe15} showed that active objects
are amenable to deductive verification, however, that work suffered
from limitations:
Asynchronous procedure calls, even with explicit suspension points,
cannot be encapsulated in a big-step abstraction using
pre-/post-conditions to describe state.  Instead, a specification must
at least partially describe the possible \emph{traces} resulting from
procedure execution, including (procedure-)\emph{internal} events,
such as synchronization, to reason about concurrency.
Because of this, the state-based approach of \cite{DinBH15,DinOwe15}
turned out to be problematic in two aspects: First, it necessitates
the use of ghost variables, in this case for recording event histories
during symbolic execution.  This, in turn, requires to reason about
histories as a data structure, hindering proof search
automation. Second, \emph{specification} of a procedure's context is
done in terms of \emph{state}.
This makes it impossible to specify the history and future wherein a
procedure is expected to operate correctly.
For example, a procedure relying on a given resource may require that
certain operations to prepare that resource were completed once it
starts. Dually, it might expect that its caller cleans up
afterwards. In particular, properties that stipulate the existence of
global traces, such as liveness, cannot be expressed with state-based
contracts. The context-aware trace contracts we define below do not
require ghost variables and they let one specify the history and
future of a called procedure.


\paragraph{Approach.}

The sketched limitations of state-based approaches to specification of
concurrent programs suggest to explore \emph{trace-based}
specification contracts.
We present \emph{context-aware trace contracts} (CATs). Syntactically,
these are formulas of a logic for symbolic traces
\cite{DBLP:journals/corr/abs-2211-09487}, generalizing first-order
pre- and post-conditions. One immediate consequence is that trace
elements, such as events, become first-class citizens and need not be
modeled with ghost variables. The possible traces of a given procedure
$m$ are specified with a CAT $C_m$ in the form of a judgment $m:C_m$,
meaning that all possible traces of $m$, \emph{including their
  context}, are described by $C_m$.

Our main contribution is to fashion these contracts as
\emph{context-aware}. This means that a CAT for a procedure $m$
consists of three parts: A generalized pre-condition describing the
assumed \emph{trace} up to the moment when $m$ starts; the possible
traces produced by executing $m$; and a post-condition that again is a
trace describing the assumed operations taking place after termination
of $m$.  The generalization of pre-/post-conditions to traces requires
careful examination of allocation of guarantees: A procedure
\emph{guarantees} only its \emph{post-state} but it \emph{assumes} the
system continues in a certain fashion---this is, however, not
guaranteed by the procedure, but by the \emph{caller}.

Let us illustrate CATs with an example. Consider a procedure
\texttt{work} that operates on a file. Its pre-condition is that the
file was opened.  Its internal specification is that the file may be
read or written to, but nothing else.  Its post-condition is that (A)
in the final state of \texttt{work}, a flag indicating that it has
finished is set, and that (B) after \texttt{work} terminated, the file
will be closed.  While the procedure assumes that the file has been
opened upon start, itself it can only guarantee (A), while it is the
caller's obligation to ensure (B).

\paragraph{Summary of Contributions and Structure.}

Our main contribution is the generalization of state-based procedure
contracts to the context-aware, trace-based CAT model that permits not
merely to specify the behavior of a procedure, but also its context.  We
apply the CAT theory to a simplified active object concurrency model,
where trace specification involves communication events, and the proof
calculus must keep track of when a procedure may start execution and when
its pre-condition must hold.

To render specification and verification practical, we add two further
ingredients: the novel concept of an \emph{observation quantifier}
lets one record the (symbolic) value of a program variable at a given
point in a CAT. This can be seen as a generalization of \texttt{old}
references in state-based contract languages \cite{JML-Ref-Manual} and
lets one compare the symbolic values of program variables at different
points within a CAT. Second, to reduce the specification and
verification effort, we generalize Liskov's behavioral subtyping
principle \cite{LiskovWing94} to CATs.

%

We first discuss the state of the art in Sect.~\ref{sec:sota}, before
we introduce our programming model in Sect.~\ref{sec:semantics}.  The
CAT concept is based on a trace logic, described in
Sect.~\ref{sec:logic}, and CATs themselves are described in
Sect.~\ref{sec:spec}.  The proof calculus is given in
Sect.~\ref{sec:calc}, before we give an example in Sect.~\ref{sec:case}, describe the Liskov principle in
Sect.~\ref{sec:liskov} and conclude in Sect.~\ref{sec:conc}.


\section{State of the Art}\label{sec:sota}

\paragraph{Traces and Contracts.}

Specifying traces in logic has a long tradition, for example, using
Linear Temporal Logic (LTL) \cite{Wolper83} for events. We focus here
on trace logics used for deductive verification.  Our use of traces
for internal specification is based on work by Bubel et
al.~\cite{DBLP:journals/corr/abs-2211-09487}, which is following a
line of research going back to Dynamic Logic with Co-Inductive
traces~\cite{BubelCHN15} that focused on a sequential while language
and is not connected to contracts.

Dynamic Trace Logic (DTL)~\cite{bern} is an extension of dynamic logic
that uses LTL formulas as post-conditions. DTL was investigated for a
fragment of Java without concurrency and implemented as a prototype in
th KeY system. 
It only specifies changes of the state and no events, and is targeting
internal behavior of methods.

ABSDL~\cite{DinBH15,DinOwe15} uses first-order logic and a ghost
variable to keep track of events. As discussed in the introduction,
the use of a ghost variable to encode trace properties in state
predicates leads to complex specification patterns.  ABSDL was
implemented in KeY-ABS~\cite{DinBH15} for the Active Object language
ABS~\cite{ABSFMCO10}.  As specification, however, it only supports
object invariants, procedure contracts must be encoded. This
specification principle does not scale to complex systems and
protocols and, as it keeps track of events, but not of states, it
cannot handle specifications for internal state change. The
specifications are also procedure-local: one cannot express global
system properties or even liveness.

Behavioral Program Logic (BPL)~\cite{DBLP:conf/tableaux/Kamburjan19}
is a parametric logic with trace-based semantics, that has been
instantiated for a calculus that supports behavioral procedure
contracts~\cite{DBLP:series/lncs/KamburjanDHJ20}.  A behavioral
procedure contract supports limited specification of context, by
specifying which procedures are allowed to run before execution
starts, but it does so in a dedicated specification pattern called
\emph{context set}, which is not uniform for trace logic.  The
approach is implemented in the Crowbar tool~\cite{crowbar} for ABS and
based on \emph{Locally Abstract, Globally Concrete (LAGC)}
semantics~\cite{DBLP:journals/corr/abs-2202-12195,DHJPT17}, a
bilayered trace semantics that differentiates between local traces of
statements and global traces of programs.  Local traces are
parameterized with a concurrent context and combined into a concrete
global trace once that context becomes known due to scheduling
decisions. In the end, a set of concrete traces is produced for a
given program.

\paragraph{Context and Contracts.}

With \emph{context} we mean the execution trace before and after a
procedure is executed, as specified from the perspective of the
procedure. In the following, we discuss approaches that view traces
from a similar perspective.  Behavioral contracts for Active Objects
are discussed above.

Session types are a typing paradigm for specification and verification
of event traces in concurrent systems~\cite{Honda08}, which was
adapted to active objects~\cite{DBLP:conf/icfem/KamburjanDC16}.
Session types have a projection mechanism that generates a local
specification from a global one.  Projection is re-interpreted as
generation of proof obligations within BPL, where \emph{global}
soundness relies on the \emph{implicit} context of the global
specification~\cite{DBLP:conf/tableaux/Kamburjan19,DBLP:conf/ifm/KamburjanC18}.

Typestate~\cite{DBLP:conf/oopsla/AldrichSSS09,DBLP:conf/ecoop/DeLineF04}
is an approach, where a trace of procedure calls is specified at the
object level. In terms of the file example in the introduction, the
order of opening, writing/reading, and closing a file would be part of
the specification of the file class.  Typestate has recently been
integrated in deductive
verification~\cite{DBLP:journals/corr/abs-2209-05136}, but is
complementary to contracts: it specifies traces from the point of view
of the entities that are the target of events, while contracts specify
traces from the point of view of the origin of events. In particular,
it presupposes object-like structures. As we can see later, this is
unnecessary for contracts, which merely require \emph{procedures}.

In this article, we focus on specifying the \emph{temporal} context of
a procedure call, i.e., the preceding and subsequent events.  Orthogonal
to this, and not our concern here, is the \emph{spatial} context,
where one specifies the relation between different parts of the heap
memory, such as in Separation
Logic~\cite{DBLP:conf/concur/OHearn04,DBLP:conf/lics/Reynolds02},
Dynamic Frames \cite{Kassios11}, Permission Logic \cite{MSS16},
etc. These are \emph{state}-based formalisms. To combine them with
CATs is future work.


\section{Program Semantics}\label{sec:semantics}

\newcommand{\acronym}{\texttt{Async}\xspace}

\subsection{The \acronym Language}
\label{sec:async-language}

\subsubsection{Syntax}

We define the \acronym language, a small, imperative language that
features asynchronous as well as synchronous procedure calls, and a
tree-like concurrency model.  It slightly simplifies typical Active
Object languages \cite{ActiveObjects17}, but it is close enough to
expose the challenges of trace-based semantics and contracts for
languages with cooperative scheduling.  We permit recursion only for
synchronous calls, which is sufficient in practice.
For the sake of being able to present relevant examples, we add a
small domain-specific language extension with file operations (second
line of statement rule in Fig.~\ref{fig:lang-syntax}).

\begin{figure}[b!t]
  $\begin{array}[t]{r@{\hspace{2pt}}r@{\hspace{2pt}}l@{\hspace{2pt}}r@{\hspace{2pt}}r@{\hspace{2pt}}l}
  P\in\mathit{Prog}&::=&\overline{M}\,\{d \ s\} & e \in\mathit{Expr}&::=& f \mid \cdots\\ 
  M\in\mathit{ProcDecl} & ::= & m()\,\{s; \ \xasync{return}\}  
  & d\in\mathit{VarDecl}  & ::= &\varepsilon \mid x; \ d 
  \\ 
  s\in\mathit{Stmt} &::= & \multicolumn{4}{l}{\xasync{skip} \mid  x = e \mid m()  \mid \, !m() \mid  \xasync{if}(e)\{s\} \mid s;s \mid }\\ 
  &&\xasync{open}(f) \mid \xasync{close}(f) \mid \xasync{read}(f) \mid \xasync{write}(f)
  \end{array}$
  \caption{Syntax of \acronym}
  \label{fig:lang-syntax}
\end{figure}

\begin{definition}[Syntax]
  The syntax of programs $P$ and their elements is given by the
  grammar in Fig.~\ref{fig:lang-syntax}.
  A program $P$ consists of a set of procedures given by $\methods{P}$
  and an init block that declares the global variables and the initial
  statement to start execution.  The global lookup table $\Gi$ is
  defined by
  $\Gi= \many{\langle m()\,\{s;\xasync{return}\} \rangle}_{m\in
    \methods{P}}\ $.
  Let $\mathsf{PVar}$ be the set of program variables with typical
  value $x$.
  Let $m$ range over procedure names and $f$ over file descriptors.
\end{definition}

There are no type annotations and no local variables in \acronym. A
procedure $M$ has a name $m$, and a procedure body for execution.
Each procedure ends with a $\returnStmt\!\!$ statement, serving merely as
a syntactic marker to simplify the semantics of process termination:
There are neither return values, nor procedure parameters, all of
which can be encoded as global variables.
A statement $s$ is either a standard imperative construct, like
assignment or sequential composition, a synchronous call $m()$, an
asynchronous call $!m()$ or a file operation for opening, closing,
reading and writing a file\footnote{We do not add the value to be
  written as a parameter, again for simplicity. This can be easily
  modelled with a global variable, if desired.}. We underspecify the
set of expressions, but require that file identifiers are literals. We
assume expression evaluation to be total.

\lstset{language=Async}

To guide the presentation and motivate our approach we use the
following running example.
\begin{example}\label{ex:files}
  The following program writes to two files using the \lstinline{do}
  procedure.  This procedure opens the file stored in the global
  variable \lstinline{file}, asynchronously issues its closing, and
  then calls \lstinline{operate}.

  \begin{lstlisting}[language=Async,numbers=none]
    do() { open(file); !closeF(); operate(); return; }
    operate() { write(file); return; }
    closeF() { close(file); return; }
    {file; file = "file1.txt"; do(); file = "file2.txt"; do(); }
  \end{lstlisting}

\subsubsection{Concurrency Model}

Before we formalize the semantics of \acronym, we point out its
cooperative tree-like scheduling for concurrency. By cooperative, we
mean that a process is preemption-free: Once a procedure starts, it
runs until the end of its code before another procedure can be
scheduled. This is standard in both Actor and Active Object languages
\cite{ActiveObjects17}. By tree-like, we mean that all asynchronously
called procedures by a process $p$, are guaranteed to run
\emph{directly after} $p$ terminates. This ensures that from the point
of view of the caller of $p$, these processes are hidden and do not
interleave with other caller processes.
\begin{example}
  \label{ex:concurr-only-async}
  Consider a program where procedure \async{m} asynchronously calls
  \async{m1} and \async{m2}, while \async{m1} asynchronously calls
  \async{m3} and \async{m4}.
  \begin{lstlisting}[language=Async,numbers=none]
    m()  { !m1(); !m2(); return }
    m1() { !m3(); !m4(); return }
    { m() }
  \end{lstlisting}

  \noindent The tree-like semantics ensures the following scheduling
  constraints:
  \begin{itemize}
  \item The processes for \async{m1} and \async{m2} run directly
    after the one for \async{m} terminates, and before any other
    process is scheduled (of a potential caller of \async{m}).
  \item The processes for \async{m3} and \async{m4} run directly
    after the one for \async{m1} terminates, and before any other
    process is scheduled.
  \item Assume that \async{m1} is scheduled before \async{m2},
    then the processes for \async{m3} and \async{m4} run before
    the one for \async{m2}.
  \end{itemize}
\end{example}

Synchronous calls are handled via inlining and are a special case of
this model.
\begin{example}
  Consider the following program:
  \begin{lstlisting}[language=Async,numbers=none]
    m1() { !m3(); return }
    m2() {...}
    m3() {...}
    { m1(); m2() }
  \end{lstlisting}

  \noindent
  Since the body of \async{m1} is inlined before the one of
  \async{m2}, under the tree-like semantics all the procedures called
  in \async{m1} (synchronous or not) run before \async{m2}.
  Therefore, is ensured that \async{m3} runs before \async{m2}.
\end{example}

\begin{example}
  The tree-like concurrency model guarantees that in
  Example~\ref{ex:files}, procedure \lstinline{operate} is executed
  before \lstinline{closeF}, because the two calls occur in the same
  scope, and the first one is synchronous while the second is
  asynchronous.
\end{example}

\subsection{States and Traces}

We define the program semantics formally, following mostly
\cite{DBLP:journals/corr/abs-2211-09487}, except for file operations,
some aspects of call identifier management, and tree-like
concurrency. First, we require some technical definitions.
%
\begin{definition}[State, State Update]
  A \emph{state} $\sigma\in\Sigma$ is a partial mapping
  $\sigma: \mathsf{PVar} \rightharpoonup \mathsf{Val}$
  from variables to values. The notation $\sigma[x\mapsto v]$
  expresses the \emph{update} of state $\sigma$ at $x$ with value $v$
  and is defined as $\sigma[x\mapsto v](y)=v$ if $x=y$ and
  $\sigma[x\mapsto v](y)=\sigma(y)$ otherwise.
\end{definition}

There is a standard evaluation function $\valDnoargs$ for expressions,
for example, in a state $\sigma=[x\mapsto0, y\mapsto1]$ we have
$\valB{\sigma}{x + y}=\valB{\sigma}{x} + \valB{\sigma}{y}=0+1=1$.

Call scopes inside events keep track of active and called processes
for procedures.  This simplifies the semantics as one does not need an
explicit process pool or stack frames.
\begin{definition}[Scope] 
   A \emph{(call) scope} is a pair $scp=(m,id)$, where $m$ is a
  \emph{procedure name} and $id$ is a \emph{call identifier}.
\end{definition}

Events keep track of side effects, in particular they keep track of
asynchronous calls, process scheduling and termination, and file
interactions.
\begin{definition}[Event Marker]\label{def:event} 
  Let $m$ be a procedure name, $id$ a call identifier, and $scp$ a
  scope.  Event markers $ev$ are defined by the
  grammar:
  \begin{align*}
    ev ::=\; & \callEvP{}{m,id} \mid \invocEvP{}{m,id} \mid \retEv(id) \mid \pushEvP{}{scp} \mid \popEvP{}{scp} \mid \\
             & \openEvP{}{f} \mid\closeEvP{}{f}\mid\readEvP{}{f}\mid \writeEvP{}{f}
  \end{align*}    
\end{definition}

We denote with $\ev{\many{e}}$ a generic event marker over expressions
$\many{e}$.  Event markers $\callEvP{}{m,id}$, $\invocEvP{}{m,id}$,
and $\retEv(id)$ are associated with a synchronous call, an
asynchronous call, and a return statement, respectively.
Event markers $\pushEvP{}{scp}$ and $\popEvP{}{scp}$ are associated
with the start (process activation after a call) and end of a
computation in a procedure (termination of a process after the return
statement has been executed) in scope $scp$, respectively.  These
events are similar to, but simpler than the ones in
\cite{DHJPT17,DinOwe15} due to the absence of futures, call
parameters, and return values.  In addition, in
\cite{DHJPT17,DinOwe15} \emph{invocation reaction} events are used
only for asynchronous calls, while we use $\pushEv$ for synchronous
calls as well to achieve greater uniformity.

%
We define dedicated event markers $\openEvP{}{f}$, $\closeEvP{}{f}$,
$\readEvP{}{f}$, and $\writeEvP{}{f}$ associated with operations on a
file $f$. These are a domain-specific extension of the framework,
added for the benefit of having examples.
\begin{definition}[Trace]
  A \emph{trace}~$\tau$ is defined by the following rules (where
  $\varepsilon$ denotes the empty trace):
  \[
    \begin{array}{l@{\;::=\;}l@{\qquad}l@{\;::=\;}l}
      \tau & \varepsilon~|~\consTr{\tau}{t} & \stateorevent & \sigma~|~\ev{\many{e}}
    \end{array}
  \]
  We define a \emph{singleton trace} as
  $\singleton{\sigma} = \varepsilon \cons \sigma$.  When an event
  $\evp{}{\many{e}}$ is generated we need to uniquely associate it
  with the state $\sigma$ in which it was generated.  To do so we
  define the corresponding \emph{event trace}
  $\evTrioPV{\sigma}{\many{e}}{}=\consTr{\langle\sigma\rangle}{\consTr{\ev{\valB{\sigma}{\many{e}}}}{\sigma}}$.
\end{definition}

Traces are finite sequences over events and states, where every event
is encapsulated in an event trace triple.  Events do not change a
state. States and events do not need to be constantly alternating,
there can be arbitrarily many state updates between the occurrence of
two events.


Sequential composition ``\lstinline|r;s|'' of statements is
semantically modeled as trace composition, where the trace from
executing \lstinline|r| ends in a state from which the execution trace
of \lstinline|s| proceeds.  Thus the trace of \lstinline|r| ends in
the same state as where the trace of \lstinline|s| begins.  This
motivates the \emph{semantic chop} ``$\chopSem$'' on
traces~\cite{HalpernShoham91,HarelKP80,NakataUustalu15} that we often
use, instead of the standard concatenation operator
``$\concatTr{}{}$''.
%
%
\begin{definition}[Semantic Chop on Traces] 
  Let $\tau_1,\,\tau_2$ be non-empty and finite traces. 
  The \emph{semantic chop} $\tau_1 \chopSem \tau_2$ is defined
  as
  $ \chopTrSem{\tau_1}{\tau_2} = \concatTr{\tau}{\tau_2} $, where
  $\tau_1=\consTr{\tau}{\sigma}$,
  $\tau_2=\concatTr{\langle\sigma'\rangle}{\tau'}$, and
  $\sigma=\sigma'$.
  When $\sigma\neq\sigma'$ the result is undefined.
\end{definition}

\begin{example}
  Let $\tau_1=\singleton{\sigma}\cons\sigma[\code{x}\mapsto1]$ and
  $\tau_2=\singleton{\sigma[\code{x}\mapsto1]}\cons\sigma[\code{x}\mapsto1,
  \code{y}\mapsto2]$, then $\tau_1 \chopSem \tau_2$ =
  $\singleton{\sigma}\cons\sigma[\code{x}\mapsto1]\cons\sigma[\code{x}\mapsto1,
  \code{y}\mapsto2]$.
\end{example}
%

Our trace semantics evaluates an individual statement
``locally''. Obviously, it is not possible to fully evaluate composite
statements in this manner. Therefore, local semantic rules perform one
evaluation step at a time and defer evaluation of the remaining
statements, which are put into a \emph{continuation}, to subsequent
rule applications. Syntactically, continuations are simply statements
$s$ wrapped in the symbol $K$. To achieve uniform definitions we
permit the case that no further evaluation is required (it has been
completed) and use the ``empty bottle'' symbol for this case.

\begin{definition}[Continuation Marker] 
  Let $s$ be a program statement, then $\contP{}{s}$ is a
  \emph{continuation marker}. The empty continuation is denoted with
  $\contP{}{\zero}$ and expresses that nothing remains to be
  evaluated.
\end{definition}

\subsection{Semantics of \acronym}
\label{sec:semantics-async}

The semantics of \acronym is two-layered: a local semantics for
small-step evaluation of a single process, and a global semantics for
the evaluation of the whole state, in particular scheduling and other
concurrency operations.

As our language is locally deterministic\footnote{Evaluation of a
  single process, in a known context.}, we define local small-step
evaluation $\valP{}{s}$ of a statement $s$ in state $\sigma$ to return
a \emph{single} trace: The result of $\valP{}{s}$ is of the form
$\concatTr{\tau}{\contF{}{s'}}$, where $\tau$ is an initial
(small-step) trace of $s$ and $\contF{}{s'}$ contains the remaining,
possibly empty, statement $s'$ yet to be evaluated.
  
To distinguish different calls of the same procedure, we generate
fresh call identifiers.  This cannot be done locally, so the most
recently used call identifier is passed as a ``counter'' $id$ to the
local evaluation rules. Yet another context parameter of the local
evaluation rules is the identifier of the currently executing scope
$\mathit{cId}$. It is passed down from a global rule. Both parameters
appear as superscripts in $\valPC{id}{\mathit{cId}}{s}$.
%
\begin{definition}[Small-Step Local Evaluation]
  The local evaluation rules defining $\valPC{id}{\mathit{cId}}{s}$
  are given in Fig.~\ref{fig:local}.
\end{definition}

\begin{figure}[bt]
  \begin{align*}
    &\valPC{id}{\mathit{cId}}{\xasync{skip}} = \  
      \concatTr{\langle\sigma\rangle}{\contF{}{\zero}} \\
    &\valPC{id}{\mathit{cId}}{\assignStmt{x}{e}}  = 
      \concatTr{\consTr{\langle\sigma\rangle}{\updatestate{\sigma}{x}{\valP{}{e}}}}{\contF{}{\zero}}\ \\ 
    &\valPC{id}{\mathit{cId}}{\xasync{return}}
      = \ \concatTr{\retEv_{\sigma}(\mathit{cId})}{\contF{}{\zero}}\\
    &\valPC{id}{\mathit{cId}}{\xasync{if}~e~\{ ~s~\}}  =  \begin{cases}
      \concatTr{\langle\sigma\rangle}{\contF{}{s}}, \ \text{if} \ \valBP{\sigma}{}{e} = \trueSem\\
      \concatTr{\langle\sigma\rangle}{\contF{}{\zero}}\, \ \text{otherwise}
    \end{cases}\\
    &\valPC{id}{\mathit{cId}}{r;s}
    =  \ \concatTr{\tau}{\contF{}{r'; s}}, \ \text{where} \ \valPC{id}{\mathit{cId}}{r}=\concatTr{\tau}{\contF{}{r'}}\ \text{and} \ \zero; s \rightsquigarrow s \\
    &\valPC{id}{\mathit{cId}}{m()} = \ \concatTr{\callEvP{\sigma}{m,id + 1}}{\contF{}{\zero}}\  \notag \\ 
    &\valPC{id}{\mathit{cId}}{!m()} = \ \concatTr{\invocEvP{\sigma}{m,id +1}}{\contF{}{\zero}} \notag\\ 
    &\valPC{id}{\mathit{cId}}{\openStmt{f}} = \ \concatTr{\openEvP{\sigma}{f}}{\contF{}{\zero}}\qquad 
      \valPC{id}{\mathit{cId}}{\closeStmt{f}} = \ \concatTr{\closeEvP{\sigma}{f}}{\contF{}{\zero}} \notag\\ 
    &\valPC{id}{\mathit{cId}}{\readStmt{f}} = \ \concatTr{\readEvP{\sigma}{f}}{\contF{}{\zero}} \qquad 
      \valPC{id}{\mathit{cId}}{\writeStmt{f}} = \ \concatTr{\writeEvP{\sigma}{f}}{\contF{}{\zero}} \notag
  \end{align*}
  \caption{Local Program Semantics}
  \label{fig:local} 
\end{figure}

The rules for asynchronous calls, synchronous calls, and return emit
suitable event traces. The difference between $\callEv$ and $\invocEv$ is
that the former directly triggers execution of a procedure in the
trace composition rules below.  In both cases a new call identifier is
generated based on $id$, therefore any two calls, no matter whether
synchronous or not, always have different call identifiers.

The rule for sequential composition assumes empty leading
continuations are discarded, the remaining rules are straightforward.
Local evaluation of a statement $s$ yields a small step $\tau$ of $s$
plus a continuation~$\contF{}{s'}$. Therefore, traces can be extended
by evaluating the continuation and stitching the result to
$\tau$. This is performed by trace \emph{composition rules} that
operate on configurations of the form $\tau,\,\cont{s}$.  To formulate
the composition rules we need to introduce auxiliary structures to
keep track of scopes and the call tree.

We define \emph{schematic traces} that allow us to characterize
succinctly sets of traces (not) containing certain events via
matching. The notation $\finiteNoEv{\many{ev}}$ represents the set of
all non-empty, finite traces \emph{without} events of type
$ev\in\many{ev}$.  Symbol $\finiteNoEv{}$ is shorthand for
$\finiteNoEv{\emptyset}$.
With $\tau_1 \finiteNoEv{\many{ev}} \tau_2$ we denote the set of all
well-defined traces $\tau_1 \chopSem \tau \chopSem \tau_2$ such that
$\tau\in\finiteNoEv{\many{ev}}$.

Maintaining the stack of call scopes is handled by the composition
rules with the help of events $\pushEv(scp)$ and $\popEv(scp)$ that
are added to the generated trace $\tau$.  To find the current call
scope in $\tau$, one simply searches for the most recent pushed scope
that was not yet popped:
\begin{definition}[Current Call Scope, Most Recent Call Identifier] 
  \label{def:lastev-currctx}
  Let $\tau$ be a non-empty trace. The \emph{current call scope} is
  defined as
  \[
    \mathit{currScp}(\tau) =
    \left
      \{
      \begin{array}{l@{\quad}l}
        scp & \tau\in\,\finiteNoEv{} \,\pushEvP{\sigma}{scp} \finiteNoEv{\pushEv,\popEv}\\[-3pt]
        \mathit{currScp}(\tau') & \tau\in\tau' \chopSem\, \pushEvP{\sigma}{scp}  \finiteNoEv{} \popEvP{\sigma'}{scp}\finiteNoEv{\pushEv,\popEv}
      \end{array}
    \right.
  \]
  The \emph{most recent call identifier} of a trace is retrieved with
  \[
    id(\tau) = \mathsf{max}\{i \mid \tau \in  \finiteNoEv{} \invocEvP{}{\_,i}\finiteNoEv{} \text{ or } \tau \in  \finiteNoEv{} \callEvP{}{\_,i}\finiteNoEv{}  \} \quad .
  \]

  If $\tau$ is not empty then we define the function
  $last(\tau)= last(\tau'\cons \sigma) = \sigma$, for some state
  $\sigma$ and possibly empty trace $\tau'$.
\end{definition}

To retrieve the most recent call identifier it suffices to consider
events that introduce fresh call identifiers, i.e.  $\callEv$ and
$\invocEv$.

To define the processes in a given trace that are eligible for
scheduling, we define the call tree of a trace that records the
dependencies of the call scopes.  Each node in the call tree is a
scope in the given trace, where an edge $(v_1, v_2)$ denotes that
$v_1$ called or invoked $v_2$.

\begin{definition}[Call Tree]
  \label{def:call-tree}
  A call tree for a trace $\tau$ is an ordered tree 
  \[(V(\tau), E(\tau), <)\] with vertices
  $ V(\tau) = \{ (m,i) \mid \tau \in  \finiteNoEv{} \invocEvP{}{m,i}\finiteNoEv{} \text{ or } \tau \in  \finiteNoEv{} \callEvP{}{m,i}\finiteNoEv{}  \}$,
  edges
  \begin{align*}
    E(\tau)  =& \{ 
                (\mathit{currScp}(\tau'),(m, i)) \mid 
                (\tau \in \tau' \chopSem \callEvP{\sigma}{m, i}\finiteNoEv{}) \\
              & \ \ \quad\quad\qquad\qquad\qquad\vee (\tau \in \tau' \chopSem \invocEvP{\sigma}{m, i}\finiteNoEv{}) \}
  \end{align*}

  \noindent and order $(m_1,i_1) < (m_2,i_2) \iff i_1 < i_2 \ .$
  We define the set of \emph{idle nodes}
  $V_{\mathit{idle}}(\tau)\subseteq V(\tau)$ as those asynchronous
  calls having not yet started to execute:
  \[
    V_{\mathit{idle}}(\tau) = \{ (m,i) \mid \tau \in \finiteNoEv{} \invocEvP{}{m,i} \finiteNoEv{\pushEvP{}{m,i}}\}
  \]
  We also define a function to retrieve all children of a given call
  scope:
  \[
    \mathit{children}(scp,\tau) = \{child \mid (scp,child) \in E(\tau)\}
  \]
\end{definition}


Observe that $V_{\mathit{idle}}$ can never coincide with $V$, since
the main scope is never idle, i.e.
$(init, 0) \notin V_{\mathit{idle}}$.

To define the tree-like semantics mentioned in
Sect.~\ref{sec:async-language}, we introduce implicit barriers for the
execution of asynchronously called procedures: a procedure that was
invoked in scope $scp$ must be scheduled before the scope $scp$ is
exited.  We do not allow pending procedure invocations in a closed
call scope.  This is formalized in the following definition.
\begin{definition}[Schedule Function]
  \label{def:schedule-function}
  Given a trace $\tau$ we define
  $$
  \mathit{schedule}(\tau) = \mathit{children}(\mathit{currScp}(\tau), \tau) \cap V_{\mathit{idle}}(\tau)
  $$
\end{definition}

The above scheduling function realizes tree-like concurrency, but can
be easily adapted to other concurrency models. For example,
$\mathit{schedule}(\tau) = \{ \mathsf{min}(V_{\mathit{idle}}(\tau))
\}$ defines a deterministic scheduler (as does any instantiation that
returns a singleton or empty set), while
$\mathit{schedule}(\tau) = V_{\mathit{idle}}(\tau)$ is the usual fully
non-deterministic scheduler.


\begin{figure}
\begin{minipage}{0.6\textwidth}
\begin{align*}
  \tau = \ & \callEvP{\sigma}{init,0}\chopSem \pushEvP{\sigma}{(init,0)} \\
  &\chopSem \callEvP{\sigma}{\mathtt{m},1}\chopSem \pushEvP{\sigma}{(\mathtt{m},1)}\\
  &\chopSem \invocEvP{\sigma}{\mathtt{m1},2}\chopSem \invocEvP{\sigma}{\mathtt{m2},3}\\
  &\chopSem \pushEvP{\sigma}{(\mathtt{m1},2)} \\
  &\chopSem \invocEvP{\sigma}{\mathtt{m3},4}\chopSem \invocEvP{\sigma}{\mathtt{m4},5}\\
  &\chopSem\retEv(2)
\end{align*}
\end{minipage}
\begin{minipage}{0.3\textwidth}
\begin{tikzpicture}[sibling distance=3em, level distance = 3em,
  nonActivated/.style={rectangle,draw,fill=black!20},
  activated/.style={rectangle,draw,fill=green!20},
  ]
  \node[activated] {main,0}
  child { node[activated] {\texttt{m},1}
  child { node[activated] {\texttt{m1},2} 
    child { node[nonActivated] (m3)  {\texttt{m3},4} }
    child { node[nonActivated] {\texttt{m4},5}}  
  }
  child { node[nonActivated] {\texttt{m2},3}}
  };
     \node (rect) at (m3) [color=blue, dashed, yshift=-.2cm, xshift=.5cm, draw,thick,minimum width=2.2cm,minimum height=1.4cm] (r) {};
     \node (rect) at (m3) [xshift=.7cm, yshift=1.3cm, draw,thick,minimum width=3cm,minimum height=4.7cm] (rBlack) {};
     \node (rect) at (r.south) [color=blue, yshift=.3cm] (r1) {\scriptsize $schedule(\tau)$};
     \node (rect) at (rBlack.north) [fill=white, draw,thick] (r) {\scriptsize Call tree of $\tau$};
\end{tikzpicture}
\end{minipage}
  \caption{Trace and call tree for an incomplete execution of program from Example~\ref{ex:concurr-only-async}}
  \label{fig:call-tree}
\end{figure}
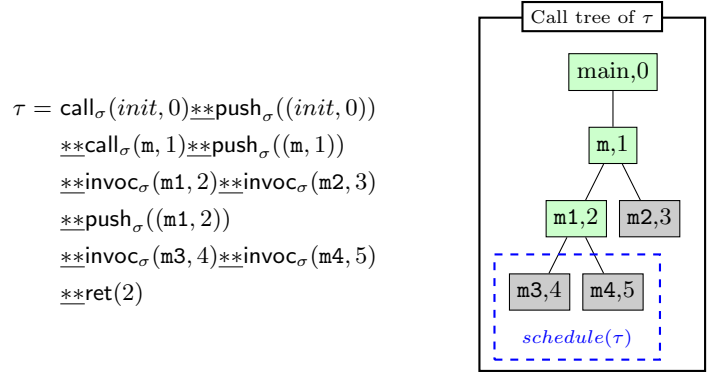

\begin{example}
  Consider an incomplete execution of the program in
  Example~\ref{ex:concurr-only-async}, with \async{m1} scheduled
  before \async{m2} and ending with the return statement of
  \async{m1}.  We show in Fig.~\ref{fig:call-tree} on the left the
  trace $\tau$ generated by the composition rules that we introduce
  later in Def.~\ref{def:composition-rules}, as well as its call tree
  on the right.  As described in Example~\ref{ex:concurr-only-async},
  \async{m3} and \async{m4} are executed before \async{m2()}. At this
  point, \async{m3} and \async{m4} can be scheduled according to the
  tree semantics, but not \async{m2}.
  \label{ex:call-tree}
\end{example}
\begin{definition}[Trace Composition Rules]
  \label{def:composition-rules}
  One global evaluation step of a configuration $\tau,\,\cont{s}$ is
  given by the \emph{trace composition rules} in
  Fig.~\ref{fig:global}.
\end{definition}

\begin{description}
\item[Rule {\normalfont$\textsc{Progress}$}.] 
  This rule uses the local small-step semantics
  $\valPC{\mathit{id}(\tau)}{i}{s}$ to evaluate $s$ in the
  continuation in the currently active process $i$. The local
  evaluation function is passed the most recent call id
  $\mathit{id}(\tau)$ and the id $i$ of the current call scope.
  It is triggered whenever the current trace does not end in a call
  event (which would require starting the execution of a synchronously
  called procedure) and no return event for the current scope has been
  yet generated (which would require scheduling of asynchronously
  called procedures).  The first two premises exclude these options,
  respectively.
  
\item[Rule {\normalfont$\textsc{Call}$}.] 
  This rule schedules the execution of a synchronously called
  procedure: If the trace ends in a call event, then the body $s'$ of
  the procedure $m$ that has just been called is \emph{prepended} to
  the current continuation $s$.
\item[Rule {\normalfont$\textsc{Run}$}.] 
  This rule schedules the execution of an asynchronously called
  procedure: If the current scope $id$ finished execution of its
  return statement (presence of return event with $id$ in $\tau$), but
  was not de-scheduled (no pop event emitted yet), and at least one of
  its remaining asynchronous calls (to $m$) is not scheduled, then the
  body $s'$ of $m$ is \emph{prepended} to the body of the continuation
  $s$ and the new scope is pushed on the current trace.
\item[Rule {\normalfont$\textsc{Return}$}.]  This rule de-schedules a
  process, which is the case if it has terminated, and has no
  asynchronous calls left to execute.
\end{description}


\begin{figure}[bt]
\gcondrule{progress-rule}{Progress}{
  \tau \not\in \finiteNoEv{}\,\callEvP{\sigma}{\_,\_} \qquad\qquad \tau \not\in \finiteNoEv{}\retEv_\sigma(i)\finiteNoEv{} \\
  \sigma = \last(\tau)  \qquad currScp(\tau)=(m,i)\qquad 
  \valPC{\mathit{id}(\tau)}{i}{s} = \concatTr{\tau'}{\cont{s'}}
}{
  \tau, \cont{s} \to \chopTrSem{\tau}{\tau'}, \cont{s'}
}

\gcondrule{call-rule}{Call}{
  \tau \in \finiteNoEv{}\,\callEvP{\sigma}{m,id}\qquad 
  \lookup{m, \Gi} = m() \ \{s';~\xasync{return}\}
}{
    \tau, \cont{s} \to \chopTrSem{\tau}{\pushEvP{\sigma}{(m,id)}}, \cont{s';~\xasync{return};~s}
}

\gcondrule{void-ret-rule1}{Run}{
  \mathit{currScp}(\tau) = (\_,id) \quad \lookup{m, \Gi} = m() \ \{s';~\xasync{return}\} \\
  \tau \in \finiteNoEv{} \retEv_\sigma(id)\finiteNoEv{\popEvP{\sigma}{\_,id}}
  \qquad
  (m,id_2) \in \mathit{schedule}(\tau) \qquad \sigma = \last(\tau)
}{
     \tau, \cont{s} \to \tau \chopSem \pushEvP{\sigma}{(m,id_2)}, \cont{s';~\xasync{return};~s}
}  

\gcondrule{void-ret-rule2}{Return}{
  \tau \in \finiteNoEv{} \retEv_\sigma(id)\finiteNoEv{\popEvP{\sigma}{m,id}}
  \qquad
  \mathit{currScp}(\tau) = (m,id)
  \quad
  \mathit{schedule}(\tau) = \emptyset
}{
     \tau, \cont{s} \to \tau \chopSem \popEvP{\sigma}{(m,id)}, \cont{s}
}  

\caption{Global Small-Step Semantics of \acronym}
\label{fig:global}
\end{figure}

We say that a rule is applied with a call identifier $i$, if it is
applied in a state $\tau,\cont{s}$ with
$\mathit{currScp}(\tau) = (\_,i)$.

Finally, we define the traces of a program, as well as the big-step
denotational semantics of statements and programs.
\begin{definition}[Program Trace]\label{def:local}
  Given a statement $s$ (with implicit lookup table) and a state
  $\sigma$, the \emph{traces} of $s$ are all the possible maximal
  sequences obtained by repeated application of composition rules
  starting from $\singleton{\sigma}, \cont{s}$.  If finite, a maximal
  sequence has the form
  $$
  \singleton{\sigma}, \cont{s} \to \cdots \to \tau, \cont{\zero}\enspace,
  $$
  also written
  $\singleton{\sigma}, \cont{s} \overset{*}{\to} \tau, \cont{\zero}$.

  Let $\tau$ be a trace and $(\_,i) = \mathit{currScp}(\tau)$. A
  maximal sequence without using the rule {\normalfont$\textsc{Run}$}
  with $i$ is denoted
  \[
    \tau, \cont{s}\overset{\times}{\to} \tau', \cont{\zero}
  \]
\end{definition}

The traces defined by $\overset{\times}{\to}$ are those, where the
asynchronously called procedures of the outermost statement $s$ are
not resolved, but those of synchronously called methods within $s$
are.

We define two program semantics, the global one represents the
execution of a program taking into account the execution of
asynchronous calls.  The local one represents the execution of a
program without doing so.
\begin{definition}[Program Semantics]
  \label{def:glob-sem}
  The \emph{global semantics} of a statement~$s$ is only defined for
  terminating statements as
  \[
    \globEval{s}{\tau}= \{\tau' \mid \tau, \cont{s} \overset{*}{\to} \tau \chopSem \tau', \cont{\zero} , \  schedule(\tau \chopSem\tau') = \emptyset \}
  \]

  The \emph{local semantics} of a statement $s$ is defined, again for
  terminating statements, as
  \[
    \locEval{s}{\tau}= \{\tau' \mid \tau, \cont{s} \overset{\times}{\to}  \tau \chopSem \tau', \cont{\zero} \}
  \]

  The semantics of a program $P$ with initial block $\{d; s\}$ is
  defined as follows, assuming no procedure is called
  ``$\mathit{init}$''. State $\sigma_d$ initializes all variables in
  $d$ to default values.
\begin{align*}
  \globEval{&P}{d}=\\ 
  &\{\tau \mid \callEvP{\sigma_d}{\mathit{init},0} \chopSem \pushEvP{\sigma_d}{\mathit{init},0}, \cont{s; \xasync{return}} \overset{*}{\to} \tau, \cont{\zero} , \  schedule(\tau) = \emptyset \}
\end{align*}

\end{definition}

The local and global semantics are connected\footnote{The split
  between local and global is inspired by the LAGC semantics for
  Active Objects~\cite{DHJPT17}.  There are some technical differences
  between our semantics and LAGC, most prominently that both our local
  and global semantics are only defined on concrete traces: We do not
  evaluate symbolically.} as follows.
Recall that we permit only synchronous recursive calls and only
terminating programs have a semantic value.

\begin{proposition}
  \label{prop:semantics-composition}
  Let $s$ be a statement without synchronous method calls, $s'$ any
  procedure body and $\tau$ a trace, then:
  \[
    \globEval{s;s'}{\tau} = \bigcup_{\tau' \in \locEval{s}{\tau}} \eval{s}{\tau}^{L} \chopSem
     \globEval{s'}{\tau \chopSem \tau'}\enspace.
  \]
  A special case is
  $\globEval{s}{\tau} = \bigcup_{\tau' \in \locEval{s}{\tau}} \eval{s}{\tau}^{L} \chopSem 
  \globEval{\zero}{\tau \chopSem  \tau'}\enspace.$
\end{proposition}

Our trace contracts are used to specify functional properties, but are
also able to express generic properties.  To keep formalities
manageable, we do not introduce exceptions and verify
exception-freedom, but merely specify what it means that the sequence
of file operations is correct.

\begin{definition}[File Correctness]\label{def:file}
  A trace $\tau$ is \emph{file correct}, if for every file descriptor
  $f$ occurring in $\tau$ and any position $i$ in $\tau$ with an event
  $\writeEvP{\sigma}{f}$, $\readEvP{\sigma}{f}$ or
  $\closeEvP{\sigma}{f}$ for some $\sigma$, there is a position $j$
  with in $\tau$ with the event $\openEvP{\sigma'}{f}$ for some
  $\sigma'$ and for no $k$ with $j<k<i$, there is a
  $\closeEvP{\sigma''}{f}$ event at position $k$ for some $\sigma''$.
\end{definition}



\section{A Logic for Trace-based Specification}\label{sec:logic}

The trace logic is based on~\cite{DBLP:journals/corr/abs-2211-09487},
however, instead of modeling program variables as non-rigid symbols in
trace formulas, we define a quantifier to bind values of program
variables to logical variables at a specific position in the trace.

\subsection{Syntax}

Formulas are constructed over a set~$\LVar$ of first-order observation
variables and a set~$\RVar$ of recursion variables.
\begin{definition}
  Let $P$ range over first-order predicates, $X$ over recursion
  variables $\RVar$, $x$ over program variables $\mathsf{PVar}$, and
  $y$ over logical variables $\LVar$.
  %
  The syntax of the logic is defined by the following grammar:
\begin{align*}
    \Phi \>::=\> & \stateFml{P} \mid X \mid \mathit{Ev} \mid
                   \Phi\wedge\Phi \mid \Phi\vee\Phi \\ 
                 & \mid \Phi\concat\Phi \mid \Phi\chop\,\Phi \mid 
                   (\muFml{X}{\Phi}) \mid \mho\,x~\mathbf{as}~y.\Phi
   \end{align*}
We forbid any occurrence of recursion variables $X$ in the scope of $\mho$.
\end{definition}

Events $\mathit{Ev}$ related to the beginning, return, and end of a
procedure $m$ in scope $(m,i)$ have the form $\startEv(m,i)$,
$\retEv(i)$, $\popEvP{}{m,i}$, respectively.
Events $\mathit{Ev}$ related to manipulation of a file $f$ have the
form $\openEvP{}{f}$, $\closeEvP{}{f}$, $\readEvP{}{f}$, or
$\writeEvP{}{f}$.
It is not necessary that events in the trace logic exactly follow the
events in the trace semantics. In fact, a certain degree of
abstraction is usually desirable. The event structure $\mathit{Ev}$ of
our logic is parameterizable. The semantics of events in the logic
relative to events in traces is defined in
Sect.~\ref{sec:semantics-logic}.

The novel aspect of our logic is the \emph{observation quantifier}
$\mho$. It addresses the problem that, unlike in state-based
Hoare-style contracts, in the asynchronous setting it is necessary to
observe the value of program variables at arbitrary points inside a
trace specification. This could be achieved with non-rigid variables,
but to control their visibility is technically complex already in the
sequential case \cite{DBLP:series/lncs/10001}. An intuitive version of
scoping is provided by quantifiers of the form
$\mho\,x~\mathbf{as}~y.\Phi$, where a \emph{logical} variable $y$
observes the value of a \emph{program} variable $x$ at the position in
the trace, where the quantifier occurs. This observed snapshot is
available within the scope $\Phi$ of the quantifier, but not outside.

\begin{example}[Notation for generic finite traces]
  The expressive power of fixed points can be used to define
  transitive closure. Let the predicate $NoEv(\many{ev})$ be true in
  any state that is not one of the events occurring in a set of events
  $\many{ev}$. We define the trace formula
  \[
    \noEvent{\many{ev}}\quad=\quad\mu
    X.\left(NoEv(\many{ev})\,\lor\,NoEv(\many{ev})\cdot X\right)
  \]
  that characterizes all non-empty, finite traces that do not contain
  an event occurring in $\many{ev}$. If $\many{ev}=\emptyset$, we
  simply write ``$\noEvent{}$''.

  Finally, we take the convention to omit writing the $\chop$ operator
  between $\noEvent{\many{ev}}$ and any adjacent trace formula. For
  example, the expression ``$\stateFml{P}\noEvent{}$'' denotes all
  finite traces that begin with a state, where $P$ is true. Another
  useful pattern is $\Phi\noEvent{\many{ev}}\Psi$, which expresses
  that any finite trace \emph{not} involving an event in $\many{ev}$
  may occur between the traces specified by $\Phi$ and $\Psi$.

  The $\noEvent{}$ notation is the syntactic equivalent of
  $\finiteNoEv{}$ defined in Sect.~\ref{sec:semantics-async} and
  extremely useful to write concise specifications. Thanks to the
  expressive power of fixed points, it is \emph{definable} in our logic.
\end{example}
  

\subsection{Semantics}
\label{sec:semantics-logic}

We use a fixed\footnote{This can be generalized as usual, if needed.}
first-order interpretation~$\interp$ for predicate and function
symbols, an environment $o:\LVar\rightarrow \mathsf{PVar}\times\Sigma$
for observation quantifiers, and a recursion variable assignment
$\rho:\RVar\rightarrow 2^{\Traces}$ that maps each recursion variable
to a set of traces.  The (finite-trace) semantics~$\evalPhiS{\Phi}$ of
formulas as a set of traces is inductively defined in
Fig.~\ref{fig:semantics-formulas}.

The observation environment $o$ records for a logic variable $y$
introduced by an observation quantifier the program variable $x$ and
program state $\sigma$ it keeps track of. Hence, for a given $y$ one
can construct a first-order variable assignment
$\beta: \mathsf{LVar} \rightarrow \mathsf{Val}$ via the following definition:
\[
\beta(o)(y) = \sigma(x)\text{ with }o(y) = (x,\sigma)
\]
\begin{figure}
  \footnotesize
  \begin{align*}
    \evalPhiS{\stateFml{P}} &= \langle\Sigma\rangle  \text{ if } \interp, \beta(o) \models P \qquad \qquad
    \evalPhiS{\stateFml{P}} = \emptyset \text{ if } \interp, \beta(o) \not\models P\\
    \evalPhiS{X} &= \rho (X)  \\
    \evalPhiS{\startEv(m,i)} &= \{\callEvP{\sigma}{m,i} \chopSem \pushEvP{\sigma}{(m,i)} \>|\> \sigma\in\States\} \\
    \evalPhiS{\retEv(i)} &= \{\retEv_{\sigma}(i)  \>|\> \sigma\in\States\} \qquad \qquad
    \evalPhiS{\popEvP{}{m,i}} = \{\popEvP{\sigma}{({m},{i})} \>|\> \sigma\in\States\} \\
    \evalPhiS{\Phi_1\wedge\Phi_2} &= \evalPhiS{\Phi_1}\cap\evalPhiS{\Phi_2} \qquad \qquad
    \evalPhiS{\Phi_1\vee\Phi_2} = \evalPhiS{\Phi_1}\cup\evalPhiS{\Phi_2} \\
    \evalPhiS{\Phi_1\concat\Phi_2} &= \{\tau_1\concat\tau_2 \>|\> \tau_1\in\evalPhiS{\Phi_1}  \wedge \tau_2\in\evalPhiS{\Phi_2}\} \\
    \evalPhiS{\Phi_1\chop\,\Phi_2} &= \{\tau_1\,\chopT\,\tau_2 \>|\> \tau_1\in\evalPhiS{\Phi_1} \wedge \tau_2\in\evalPhiS{\Phi_2}\} \\
\evalPhiS{\mho\, x~\mathbf{as}~y.\Phi} &= \{\langle\sigma \rangle \chopSem \tau ~|~ \tau \in \evalPhiNoArg{\Phi}_{\rho, o[y \mapsto (x,\sigma)]} \wedge \sigma \in \Sigma\}\\
    \evalPhiS{\muFml{X}{\Phi}} &= \sqcap\, \{ F  \>|\> \evalPhi{\Phi}{}{\rho [X \mapsto F]}{o} \sqsubseteq F \}
  \end{align*}
  \caption{Semantics of Trace Formulas}
  \label{fig:semantics-formulas}
\end{figure}
Our language does not include equality over logical observation
variables.
%
\begin{definition}
  The semantics of our logic is given in
  Fig.~\ref{fig:semantics-formulas}, where $\sqsubseteq$ and $\sqcap$
  denote point-wise set inclusion and intersection, respectively, and
  $\interp, \beta \models P$ is first-order satisfiability under
  interpretation $\interp$ and variable assignment $\beta$.
\end{definition}

By a \emph{trace formula} we mean a formula of our logic that is
closed with respect to both first-order and recursion variables.
Since the semantics of a trace formula does not depend on variable
assignments, we sometimes use $\evalPhiNoArg{\Phi}$ to denote
$\evalPhiS{\Phi}$ for arbitrary $\rho$ and $o$.

\begin{example}
  The following formula describes all traces, where the value of
  program variable $x$ decreases after a call to $\mathtt{decr}$ with
  call identifier $1$:
  \[
    \noEvent{} \mho\, x~\mathbf{as}~y_1.\left( \startEv(\mathtt{decr},1)\noEvent{}\retEv(1)\chop\,\mho\, x~\mathbf{as}~y_2.\stateFml{y_1 > y_2}\right)\noEvent{}
  \]
\end{example}


\section{Contracts}\label{sec:spec}

\subsection{The Concept of Trace-aware Contracts}

Our goal is to generalize contracts, where the pre- and
post-conditions are state formulas, to contracts, where initial and
trailing \emph{traces} may occur.

For example, consider a procedure $\mathtt{operate}$ that works on a
file. First it prepares the file in some way, then it computes
something with the data in it, finally it tidies the file up. In
addition, $\mathtt{operate}$ assumes the file was opened before it
starts and that someone takes care to close it after it finishes.

The internal actions of $\mathtt{operate}$ are specified as a
\emph{trace} over suitable events.  However, its pre- and
post-conditions are also \emph{traces}, as they do not specify the
moment \emph{when} the procedure is called and when it terminates, but
operations performed \emph{sometime before} and \emph{after}.  The
\emph{global} trace has the following shape:
\begin{equation}
  \mathit{assume}\chop\overbrace{\mathit{work}}^{\mathtt{operate}}\chop\mathit{continue}\label{eq:1}  
\end{equation}

In the following we abbreviate the trace formulas with their first
letter, that is $a$ for $\mathit{assume}$, etc.
\emph{State} pre- and post-conditions are the states where the assume
and work traces (work and continue) overlap, so we can
refine~(\ref{eq:1}) into:
\begin{equation}
  \mathit{assume}\chop\stateFml{Pre}\chop\mathit{work}\chop\stateFml{Post}\chop\mathit{continue}\label{eq:2}  
\end{equation}

Formula $\stateFml{Pre}$ describes the states a caller must be in when
$\mathtt{operate}$ is called, which are the states that
$\mathtt{operate}$ can assume to be started in. Dually,
$\stateFml{Post}$ describes the states $\mathtt{operate}$ ensures upon
finishing, which are the states the caller expects to be in, after the
call to $\mathtt{operate}$. We say that the trace formulas assume and
continue are the \emph{context} described by the above trace
specification.  Contexts pose a challenge to modularity: For a
Hoare-style contract, the caller is responsible for the pre-condition,
while the callee is responsible for the post-condition. The temporal
dimension of the procedure (first call, then termination) and the
temporal dimension of the contract (first pre-condition, then
post-condition) coincide.  This is not the case for the trace
contexts: The trace \textit{continue} occurs \emph{after} the call,
but must be established by the caller.  Moreover, the context is not
local to the call site of the procedure, it describes arbitrary
actions before and after the call.

Consider Fig.~\ref{fig:callcall}, with procedure
$\mathtt{m}_\mathtt{orig}$, containing a synchronous call to
\texttt{m}, which in turn contains another synchronous call to
\texttt{ma}.  From the perspective of \texttt{ma} the post-condition
($\theta_{c_\mathtt{ma}}$, using the notation we introduce in
Def.~\ref{def:contract}) describes the actions of the caller
(procedure \texttt{m}) \emph{and the complete call stack}, i.e., the
caller's callers such as $\mathtt{m}_\mathtt{orig}$, which are unaware
of the call to \texttt{ma}.

\begin{figure}
  \includegraphics[width=\textwidth]{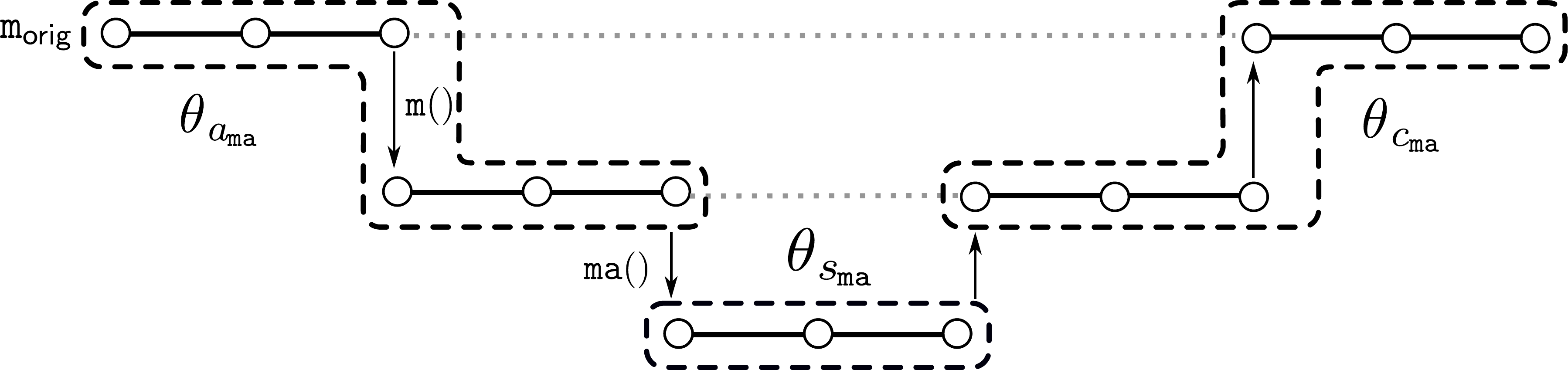}
  \caption{Scope of the post-condition}
  \label{fig:callcall}
\end{figure}

Before we turn our attention to the solution of these difficulties in
verification, let us formalize the syntax and semantics of a trace
contract.
%
%
%


\subsection{Formal Trace Contracts}

\begin{definition}[Trace Contract,
  Pre-/Post-Trace]\label{def:contract}
  Let $\mho_{\many{x,y}}$ denote a possibly empty list of
  observational quantifiers over program variables $\many{x}$ and
  logic variables $\many{y}$.  A \emph{trace contract} $C_m$ for a
  procedure $m$ has the form
  \[
    C_m = \assumeAndContinue{\theta'_{a_m} \concat \mho_{\many{x_1,y_1}}.\stateFml{q_{a_m}}}{\stateFml{q_{a_m}} \concat \, \theta'_{s_m} \concat \mho_{\many{x_2,y_2}}.\stateFml{q_{c_m}}}{ \stateFml{q_{c_m}}\concat \, \theta'_{c_m}}\enspace,
  \]
  where $\stateFml{q_{a_m}}$, $\stateFml{q_{c_m}}$ are first-order
  predicates\footnote{$\stateFml{q_{a_m}}$ and
      $\stateFml{q_{c_m}}$ correspond to $\stateFml{Pre}$ and
      $\stateFml{Post}$ above. Of course, it is redundant that these
      formulas occur twice in $C_m$, but we want each part of a
      trace contract to be readable on its own.} and $\theta'_{a_m}$,
    $\theta'_{s_m}$, $\theta'_{c_m}$ are trace formulas.
    We call
    $\theta_{a_m} = \theta'_{a_m} \concat
    \mho_{\many{x_1,y_1}}.\stateFml{q_{a_m}}$ the \emph{pre-trace},
    $\theta_{s_m} = \stateFml{q_{a_m}} \concat \, \theta'_{s_m}
    \concat \mho_{\many{x_2,y_2}}.\stateFml{q_{c_m}}$ the
    \emph{internal behavior}, and
    $\theta_{c_m} = \stateFml{q_{c_m}}\concat \, \theta'_{c_m}$ the
    \emph{post-trace} of the contract.

    We impose the following restrictions that express that all
    observation variables in the pre-trace are bound, all free logical
    variables in the internal behavior are bound by the observation
    variables of the pre-trace, and all free logic variables in the
    post-trace are bound by the observation variables of the pre-trace
    or internal behavior:
  \begin{itemize}
  \item $\mathsf{fv}(\theta_{a_m}) = \emptyset$
  \item $\mathsf{fv}(\theta_{s_m}) \subseteq \many{y_1}$
  \item $\mathsf{fv}(\theta_{c_m}) \subseteq \many{y_1} \cup \many{y_2}$
  \end{itemize}
\end{definition}

A possible contract of procedure $\mathtt{operate}$ in
Example~\ref{ex:files} is as follows.  It expects the file was opened,
has not been closed or opened again, and has not been written to yet.
Then $\mathtt{operate}$ ensures not to open or close it, abstracting
away from the actual work.  Finally, the contract stipulates that the
file will be closed by one of the callers, while not having been
opened, closed, or written to.
\[
  \assumeAndContinue{\noEvent{}\mho\, \mathtt{file}~\mathbf{as}~f.~\opf  \noEvent{\opf, \clf, \workf}}{\noEvent{\opf, \clf }}{\noEvent{\opf, \clf, \workf}\clf \noEvent{}}
\]
\end{example}

We classify trace contracts according to the trace formulas they
contain.  A contract is context-aware if it has a non-trivial pre- or
post-trace:
%
\begin{definition}[Context-Aware Contract]
  Let $C_m$ be a trace contract as in Def.~\ref{def:contract}:
  \[
    C_m = \assumeAndContinue{\theta'_{a_m} \concat \mho_{\many{x_1,y_1}}.\stateFml{q_{a_m}}}{\stateFml{q_{a_m}} \concat \, \theta'_{s_m} \concat \mho_{\many{x_2,y_2}}.\stateFml{q_{c_m}}}{ \stateFml{q_{c_m}}\concat \, \theta'_{c_m}\enspace,}
  \]
  Contract $C_m$ is \emph{context-aware} if at least one of
  $
    \theta'_{a_m} \not\equiv\,\noEvent{}$
    or $\theta'_{c_m} \not\equiv\,\noEvent{}
  $ holds.
  Contract $C_m$ is a \emph{state contract} if 
  $
    \theta'_{s_m} \equiv\,\noEvent{}
  $.
  Otherwise, it is a \emph{proper trace contract}.
  
\end{definition}

Thus, a Hoare-style contract is a non context-aware state contract,
while the contracts in~\cite{DBLP:journals/corr/abs-2211-09487} are
non context-aware proper trace contracts.
The previously introduced cooperative
contracts~\cite{DBLP:series/lncs/KamburjanDHJ20} are context-aware
state contracts, however, with a non-uniform treatment of the context.

As a final note, before we turn to the technical machinery behind
trace contracts, we stress that they naturally extend to asynchronous
communication.  Consider Fig.~\ref{fig:acallcall}, with procedure
$\mathtt{m}_\mathtt{orig}$, containing a synchronous call to
\texttt{m}, which now contains two \emph{asynchronous} calls to
\texttt{ma} and \texttt{mb}.  Analogous to Fig.~\ref{fig:callcall},
the post-trace of \texttt{ma} ($\theta_{c_\mathtt{ma}}$) describes the
actions of the caller (\texttt{m}) and the complete call stack,
including the \emph{asynchronous} callers, and subsequently running
methods such as \texttt{mb}.

\begin{figure}
\includegraphics[width=\textwidth]{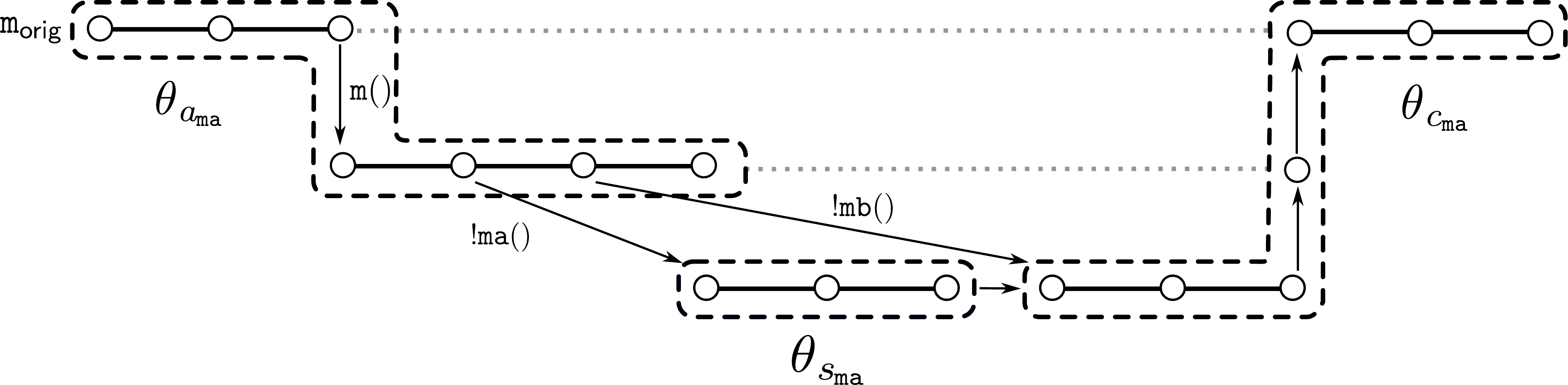}
\caption{Scope of the post-trace with asynchronous calls}
\label{fig:acallcall}
\end{figure}

The same setup from \texttt{m}'s point if view is shown in
Fig.~\ref{fig:bcallcall}: Its inner specification contains the traces
of the methods it asynchronously calls.

\begin{figure}
\includegraphics[width=\textwidth]{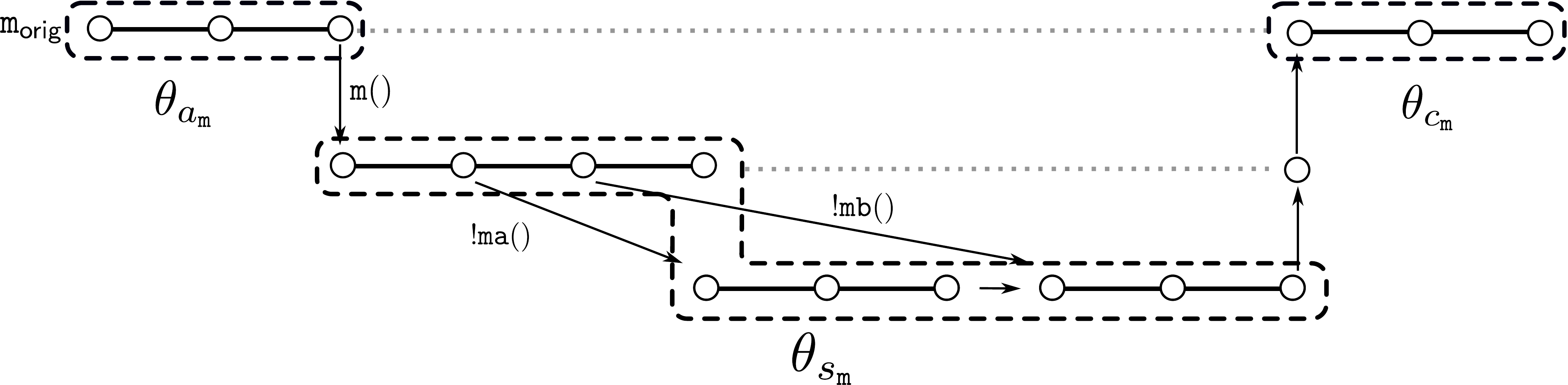}
\caption{Scope of the inner trace with asynchronous calls}
\label{fig:bcallcall}
\end{figure}

\subsection{Events versus Predicates}

The style of specification in Example~\ref{ex:files} relies on $\opf$
and $\clf$ being events---if they were predicates, then $\opf$ would
be evaluated as ``it is true that $f$ is open''.  The following
contract cannot be fulfilled, because $isOpen(f)$ cannot be true and
false at the same time ($f$ is a \emph{logical}, rigid variable):
\[
  \assumeAndContinue{\noEvent{}\mho\, \mathtt{file}~\mathbf{as}~f.~\stateFml{isOpen(f)}}{\stateFml{isOpen(f)}\noEvent{}\stateFml{!isOpen(f)}}{\stateFml{!isOpen(f)} \noEvent{}}
\]

If we want to express without events that $f$ was open before and
is closed later, then we need to introduce a second observation.
\begin{align*}
  \assume{&\noEvent{}\mho\, \mathtt{file}~\mathbf{as}~f.~\stateFml{isOpen(f)}}\\
          &\mid~{\stateFml{isOpen(f)}\noEvent{}\mho\, \mathtt{file}~\mathbf{as}~f'.~\stateFml{!isOpen(f')}}\\
          &\mid~\continue{\stateFml{!isOpen(f')} \noEvent{}}
\end{align*}

This contract still lacks the information that $f$ and $f'$ refer to
the same file. This can be addressed with a function $\mathsf{id}$
that retrieves the id of a file descriptor.
\[
  \assumeAndContinue{\noEvent{}\mho\, \mathtt{file}~\mathbf{as}~f.\stateFml{isOpen(f)}}{\noEvent{}}{\mho\, \mathtt{file}~\mathbf{as}~f'.\stateFml{\mathsf{id}(f) \doteq \mathsf{id}(f')\, \wedge\, !isOpen(f')} \noEvent{}}
\]

\subsection{Semantics of Trace Contracts}

We formalize what it means for a contract to hold for a given
procedure $m$.  Intuitively, a contract for $m$ holds in a global
trace $\tau$, if it the trace can be chopped along the events related
to $m$, such that each part of the contract holds.  We formalize this
intuition as \emph{trace adherence}.
%
\begin{definition}[Trace Adherence]\label{def:tadhere}
  Let $C_m$ be a contract for procedure $m$ with
  \[
    C_m = \assumeAndContinue{\theta'_{a_m} \concat \mho_{\many{x_1,y_1}}.\stateFml{q_{a_m}}}{\stateFml{q_{a_m}} \concat \, \theta'_{s_m} \concat \mho_{\many{x_2,y_2}}.\stateFml{q_{c_m}}}{ \stateFml{q_{c_m}}\concat \, \theta'_{c_m}}.
  \]
  %
  %
  We say that trace $\tau$ \emph{adheres} to $C_m$ for call identifier
  $i$ of $m$ if
  %
  \begin{align*}
    \tau \in \evalPhiNoArg{
    \theta'_{a_m} \concat \mho_{\many{x_1,y_1}}.\bigl(&\stateFml{q_{a_m}} \concat \mathsf{start}(m,i) \concat
    \stateFml{q_{a_m}}\concat \theta'_{s_m} \concat\\
    & \mho_{\many{x_2,y_2}}.\left(\stateFml{q_{c_m}} \concat \mathsf{pop}(i) \concat
    \stateFml{q_{c_m}}\concat \, \theta'_{c_m}\right)\bigr)
}_{\emptyset,\emptyset}
  \end{align*}
  We write this as $\tau,i \models C_m$
\end{definition}

\begin{definition}[Procedure Adherence, Program Correctness]\label{def:adhere}
  Let $P=\many{m}~\{d~s\}$ be a program containing a procedure
  $m \in \many{m}$ with contract $C_m$.
  Given a trace $\tau$, let $\mathsf{idOf}(m,\tau)$ be the set of all
  call identifiers in $\tau$ occurring in a call scope with $m$.  We
  say that $m$ \emph{adheres to} $C_m$ in $s$, written
  $m \models C_m$, if
  \begin{align*}
    \forall \tau \in \globEval{P}{d}.~\forall i \in \mathsf{idOf}(m,\tau).~\tau,i \models C_m\enspace.
  \end{align*}
  Consider the init block of program $P$ as an implicitly declared
  procedure and part of $\many{m}$, so it is uniformly handled in
    $\many{m}$. Then $P$ is \emph{correct}, written $\models P$, if
  \[
    \forall m\, \in \many{m}.~m \models C_m\enspace.
  \]
\end{definition}

Defs.~\ref{def:tadhere}, \ref{def:adhere} are based on the
$\mathsf{pop}$ event, not $\mathsf{ret}$, meaning that the call
  scope is completed. In consequence, they specify the final state of
a procedure and \emph{all of its asynchronously called} procedures.

Context-aware contracts preserve a decisive degree of modularity for
verification, because a contract can still replace inlining the
procedure body during verification---and this is crucial for
inter-procedural verification to scale. However, as shown in the next
section, one needs additional machinery to keep track of the
post-trace.
To make this precise, we need a weaker form of adherence: A procedure
\emph{weakly} adheres to its contract, if it adheres to its pre-trace
and the specification of its inner behavior.  Later we show that weak
adherence of all procedures, together with the design of the proof
system, ensures (strong) adherence.
\begin{definition}[Weak Adherence]
  Let $C_m$ be the contract of a procedure $m$ as in
  Def.~\ref{def:contract}.  Let $\widehat{C_m}$ be the contract that
  is like $C_m$, except $\theta'_{c_m}=\,\noEvent{}$.
  We say that $m$ \emph{weakly adheres to} $C_m$, if it adheres to
  $\widehat{C_m}$.
\end{definition}


\section{Proof Calculus}\label{sec:calc}

The proof calculus for context-aware contracts is based on symbolic
program execution, while tracking the symbolic trace with the help of
\emph{updates}~\cite{DBLP:series/lncs/10001}. In contrast to previous
work~\cite{DBLP:journals/corr/abs-2211-09487}, here we perform
\emph{eager} symbolic execution, but \emph{lazy} weakest precondition
computation: The rules for symbolic execution connect an update
prefix, the executed program, and a trace specification, and
manipulate all three of them. This simplifies rules for trace logics
that specify behavioral
properties~(cf.~\cite{DBLP:conf/tableaux/Kamburjan19}).

\subsection{Trace Updates}
\label{sec:trace-updates}

During symbolic execution we eagerly generate \emph{trace updates}
consisting of state updates and event updates, to keep track of the
trace in terms of changes of the program state and of the generated
events.
One can think of a trace update as a sequence of symbolic state
changes and symbolic events that resulted from symbolic execution of
the program under verification until the current state.

\begin{definition}[Trace Update]
  The syntax of trace updates is defined by the following grammar,
  where $\epsilon$ is the empty sequence of updates
  \begin{align*}
    u ::=~& \upl v:= e\upr ~|~ \upP{\mathsf{Ev}(\overline{e})} \\
    \update  ::=~&\epsilon~|~u\update
  \end{align*}
  \begin{figure}
    \lstset{basicstyle=\footnotesize}
    \begin{align*}
      & \valP{\mathit{id}, \mathit{cId}}{u\update s}
        =\concatTr{T}{\contF{}{\update s}}, \ \text{if} \ \valP{\mathit{id}, \mathit{cId}}{u}=\concatTr{T}{\contF{}{\zero}}\\
      & \valP{\mathit{id}, \mathit{cId}}{\upP{v:=e}} = \{\singleton{\sigma} \cons \sigma[v\mapsto e]\} \concat \contF{}{\zero}\\
      & \valP{\mathit{id}, \mathit{cId}}{\upP{\runEv(m,i,{sy})}} = \callEvP{\sigma}{m,i}\chopSem\eval{\xasync{s; return}}{\callEvP{\sigma}{m,i}}^G\concat \cont{\zero},\\
      & \qquad\qquad\qquad\text{ where }\mathsf{lookup}(m,\mathcal{G}) = m()\, \{\xasync{s; return}\}\\
      & \valP{\mathit{id}, \mathit{cId}}{\upP{\runEv(m,i,{as})}} = \eval{\xasync{s; return}}{\callEvP{\sigma}{m,i}}^G\concat \cont{\zero},\\
      & \qquad\qquad\qquad\text{ where }\mathsf{lookup}(m,\mathcal{G}) = m()\, \{\xasync{s; return}\}\\
      & \valP{\mathit{id}, \mathit{cId}}{\upP{\invocEv(m,i)}} = \{\invocEv_\sigma(m,i)\}\concat \cont{\zero} \\ 
      & \valP{\mathit{id}, \mathit{cId}}{\upP{\startEv(m,i)}} = \{\callEvP{\sigma}{m,i}\, \chopSem\, \pushEvP{\sigma}{(m,i)}\} \concat \cont{\zero}\\ 
      & \valP{\mathit{id}, \mathit{cId}}{\upP{\retEv(i)}} =  \{\retEv_{\sigma}(i)\} \concat \cont{\zero}\\
      & \valP{\mathit{id}, \mathit{cId}}{\upP{\popEvP{}{m,i}}} =  \{\popEvP{\sigma}{(m,i)}\} \concat \cont{\zero}
    \end{align*}
    \caption{Local evaluation of statements with update prefixes}
    \label{fig:stmt-update-eval}
  \end{figure}
\end{definition}

The local evaluation rules for statements $\xasync{s}$ prefixed with
updates $\update$ are given in Fig.~\ref{fig:stmt-update-eval}.  In
contrast to the local semantics, here the returned configuration
consists of a \emph{set of} traces and a continuation marker.  The
reason is that in the calculus we abstract away from the execution of
procedure calls by means of procedure contracts.  Therefore, the body
of a called procedure is never actually inlined, so its evaluation
cannot be local, but it must be global, resulting in a \emph{set of}
traces representing every possible execution of the procedure body
which may contain asynchronous calls. Here we overload the $\chopSem$
operator to work on sets of traces.
The non-obvious rules are those for $\runEv(m,i,{sy})$ and
$\runEv(m,i,{as})$, where $sy$ and $as$ are literals that mark the
call mode.  Run events represent the set of all possible executions
of a procedure in the context $(m,i)$ after synchronous ($sy$) and
asynchronous ($as$) scheduling, respectively.  In contrast to the
semantics of $\runEv(m,i,{as})$, the one of $\invocEv(m,i)$ is simply
a trace event that keeps track of the asynchronous invocation of $m$.

With these rules we integrate the semantics of updates directly into
the local semantics, extending Def.~\ref{def:local}, via
\[
  \eval{\update}{\tau}= \{\tau' \mid \tau, \cont{\update} \overset{\times}{\to}  \tau \chopSem\tau', \cont{\zero}\}\enspace,
\]
where, the \textsc{Progress} rule is generalized such that, given
$\valP{\mathit{id}, \mathit{cId}}{\update s} = T \concat K(s')$, any
of the traces $\tau' \in T $ is considered for extending the current
trace.

\subsection{Judgments}

We define a sequent calculus based on four different forms of judgment
that express the connection between updates, statements and trace
formulas. A sequent has the form $\Gamma \vdash \gamma$, where
$\gamma$ is one of the judgements defined below, and $\Gamma$ a set of
such judgements. $\Gamma \vdash \gamma$ means: It holds for all states
$\sigma$ that if $\sigma$ is a model for all formulas in $\Gamma$,
then it is also a model for $\gamma$.

The first judgment form expresses that a procedure weakly adheres to
its contract.  When it occurs as a premise, it can be used to assume
the contract during symbolic execution. When it occurs in a
conclusion, it is the proof obligation that must be established by
verification. This judgment is independent of a given state.
\begin{definition}[Contract Judgment]
  The judgment $m: C_m$ expresses that a procedure $m$ weakly adheres
  to its contract $C_m$:
  \[
    \models m : C_m \iff m \models \widehat{C_m}
  \]
\end{definition}

We cannot use (strong) adherence here, because a procedure cannot
control, whether its caller guarantees the post-trace.  Strong
adherence will emerge later as a global property of the calculus.

The second judgment form expresses that an update $\update$ describes
traces that are models for some trace formula $\Phi$. It is relative
to a state $\sigma$.
\begin{definition}[Local Update Judgment]
  \label{def:local-update-judgement}
  The judgment $\judge{\update}{\Phi}$ holds in a state $\sigma$ if
  all traces of $\update$, prefixed with $\sigma$ are specified by
  $\Phi$:
  \[
    \sigma \models \judge{\update}{\Phi}
    \iff
    \eval{\update}{\langle\sigma\rangle} \subseteq \eval{\Phi}{}
  \]
\end{definition}

The next two judgments forms are similar to the previous ones, but
have a \emph{global} nature in the sense that they do not only
consider the traces described by the local semantics of a statement or
update, but also include the traces generated by calls in the updates
and statements.  This means that any unresolved \invocEv event leads
to execution of the associated procedure.
%
\begin{definition}[Update and Statement Global Judgments]
  \label{def:judgement-global-update}
  The global judgments $\GlobalJudge{\update s}{\Phi}$ and
  $\GlobalJudge{\update}{\Phi}$ hold in a state $\sigma$ if all traces
  described by $\update$ and $s$, respectively, and starting in
  $\sigma$ are specified by $\Phi$:
  %
  %
  \begin{align*}
    \sigma \models \GlobalJudge{\update s}{\Phi}
    & \iff \bigcup_{\tau' \in \eval{\update}{\singleton{\sigma}}}
      \eval{\update}{\langle\sigma\rangle}  \chopSem \eval{s}{\tau' }^{G} \subseteq \eval{\Phi}{}\\
    \sigma \models \GlobalJudge{\update}{\Phi}
    & \iff \bigcup_{\tau' \in \eval{\update}{\singleton{\sigma}}}
      \eval{\update}{\langle\sigma\rangle} \chopSem \eval{\zero}{\tau'}^{G}  \subseteq \eval{\Phi}{}
  \end{align*}
\end{definition}

A subtle point is that $\eval{\zero}{\tau'}^{G}$ still generates the
traces of the non-resolved procedures in $\tau'$, even though it
executes the empty context.

Please observe that the global judgements are global with relative to
their scope, not the whole program.  Regarding the difference between
local and global update judgments, consider update
$\update = \upl \invocEv(m,i)\upr\upP{\retEv(i')}$ and trace
specification
$\Phi =\,
\noEvent{}\invocEv(m,i)\concat\retEv(i')\concat\startEv(m,i)\noEvent{}$. Locally,
$\Phi$ is not a valid specification for $\update$, because for any
$\sigma$:
\begin{align*}
  & \sigma\models \judge{\upl \invocEv(m,i)\upr \upP{\retEv(i')}}{\,\noEvent{}\invocEv(m,i)\concat\retEv(i')\concat\startEv(m,i)\noEvent{}}\\
  \iff &
         \eval{\upl \invocEv(m,i)\upr\upP{\retEv(i')}}{\langle\sigma\rangle} \subseteq \eval{\noEvent{}\invocEv(m,i)\concat\retEv(i')\concat\startEv(m,i)\noEvent{}}{\sigma}\\
  \iff &
         \{\invocEvP{\sigma}{m,i} \chopSem\,\retEv_{\sigma}(i')\} \subseteq \eval{\noEvent{}\invocEv(m,i)\concat\retEv(i')\concat\startEv(m,i)\noEvent{}}{\sigma}
\end{align*}

This is obviously not the case. However, it is true as a global
judgment, because in this case the invocation event starts procedure
$m$, which adds the required event:
\[
  \sigma \models \GlobalJudge{\upl \invocEv(m,i)\upr \upP{\retEv(i')}}{\,\noEvent{}\invocEv(m,i)\concat\retEv(i')\concat\startEv(m,i) \noEvent{}}
\]

We override the $schedule$ function of
Def.~\ref{def:schedule-function} to accept also trace updates.
\begin{definition}[Schedule Function over Trace Updates]
  \begin{align*}
    schedule(\update) = \{(m,i)\mid\, & \update = \update_1\upP{\invocEvP{}{m,i}}\,\update_2\,\upP{\retEv(oId)}\,\update_3\\ 
                        & \text{where } \update_3 \text{ does not contain } \runEv(m,i,as)\}  
  \end{align*}
\end{definition}

\begin{lemma}[Soundness of Scheduling]
  \label{lemma:soundness-scheduling}
  \[
    \mathit{schedule}(\update) = \bigcup_{\tau \in \eval{\update}{}} \mathit{schedule}(\tau)
  \]
\end{lemma}

The proof is in the appendix.  The tree-like concurrency model is
crucial: The traces in the semantics of $\mathcal{U}$ contain
\emph{more} call scopes,
but the tree-like concurrency model ensures that
they are all closed, and hence cannot be scheduled.

\subsection{Proof Rules}
\label{sec:proof-rules}

We have four classes of proof rules:
\begin{enumerate*}[label=(\arabic*)]
\item for procedure contract judgments,
\item for symbolic execution of the sequential part of the language,
  i.e., straight-line programs and synchronous calls under the global
  judgment,
\item for asynchronous calls under the global judgment, and 
\item to reduce updates in a global judgment to a local judgment.
\end{enumerate*}
We do not consider the rules for updates in local judgments.  These
are an open research question orthogonal to context-aware
contracts. All given rules for procedure calls use contracts, we
  refrain from providing rules that resolve calls by inlining.

\paragraph{Procedure Contracts.}

The rule for procedure contracts requires to prove the global judgment
that the body of a procedure, starting when the pre-trace holds,
generates a trace where the pre-trace, the internal behavior and the
state post-condition hold.  Being global, this trace includes all
asynchronously called procedures.  To model the trace at the moment
when the procedure starts, a fresh, uninterpreted update symbol
$\mathcal{V}$ is used as the update prefix to describe \emph{locally}
the pre-trace.

The observation variables $\many{y}$ are skolemized into constant
symbols $\many{c^y}$.  This substitution is denoted
$[\many{y} \setminus \many{c^y}]$. Notation $\inline(m)$ retrieves the
body of procedure $m$. Judgments $C$ must all be contract
judgments. 
For simplicity, and without loss of generality, we set the identifier
of the current scope to a fixed constant $\mathit{oId}$.
\[
  \seqRule{Contract}{%
  \mathcal{V},\many{c^y_1}, \many{c^y_2}, \mathit{oId}~\text{fresh}\qquad
  \theta_\mathit{pre} = \left(\theta'_{a_m} \concat \stateFml{q_{a_m}}\right)\left[\many{y_1} \setminus \many{c^y_1}\right]\\
  \theta_\mathit{post} = \left(\theta'_{a_m} \concat \stateFml{q_{a_m}}\concat \theta_{s_m}' \concat \stateFml{q_{c_m}}\right)\left[\many{y_1} \setminus \many{c^y_1}\right]\left[\many{y_2} \setminus \many{c^y_2}\right]\\
    \seq{C, \judge{\mathcal{V}\upl\startEv(m, \mathit{oId})\upr}{\theta_\mathit{pre}}}{\GlobalJudge{\mathcal{V}\upl\startEv(m, \mathit{oId})\upr\,\inline(m)}{\theta_\mathit{post}}}}
  {\seq{C}{\judge{m}{\assumeAndContinue{\theta'_{a_m} \concat \mho_{\many{x_1,y_1}}.\stateFml{q_{a_m}}}{\stateFml{q_{a_m}} \concat \, \theta_{s_m}' \concat \mho_{\many{x_2,y_2}}.\stateFml{q_{c_m}}}{ \stateFml{q_{c_m}}\concat \, \_}}}}
\]

It is worth pointing out that the post-trace $\theta'_{c_m}$ in the
contract (see Def.~\ref{def:contract}) is not used here. Indeed it is
out of the scope of the procedure and occurs as an additional
obligation to the \emph{caller} in the rules below.  It is essential
that $\stateFml{q_{c_m}}$ does not describe the state after the final
statement of $m$, rather it describes the final state after the last
asynchronously called procedure terminates.

\paragraph{Symbolic Execution of Straight-Line Programs.}

The schematic rules of the symbolic execution calculus for
straight-line programs are in
Fig.~\ref{fig:straight-line-programs}, with the core language on
  the left and the domain-specific extension on the right.  The
\textsf{Assign} rule generates a state update for the assigned
variable.
Rule \textsf{Cond} branches a proof according to guard $e$ and rule
\textsf{Return} generates the eponymous event update.  The rules for
operations on a file $f$ generate the associated events. The rules
\textsf{Close}, \textsf{Read}, and \textsf{Write} require that the
file was opened and not closed intermittently.

\begin{figure}
  \begin{minipage}{0.49\textwidth}
    \[
      \seqRule{\textsf{\small{Assign}}}{
        \sequent{}{\update \upl v:= e\upr\GlobalJudge{s}{\Phi}}
      }{
        \sequent{}{\update \, \GlobalJudge{v=e\xasync{;}s}{\Phi}}
      }\]
    \[
      \seqRule{\textsf{\small Cond}}{
        \sequent{\update:\noEvent{}\stateFml{e}}{\GlobalJudge{\update s\xasync{;} s'}{\Phi}} \\
        \sequent{\update:\noEvent{}\stateFml{!e}}{\GlobalJudge{\update s'}{\Phi}}
      }{
        \sequent{}{\update \ \GlobalJudge{ \xasync{if} \, e \, \{s\}\xasync{;} s'}{\Phi}}
      }
    \]
    \[
      \seqRule{\textsf{\small Return}}{
        \sequent{}{\mathcal{U} \upl\retEv(oId)\upr\GlobalJudge{}{\Phi}}
      }{
        \sequent{}{\mathcal{U}\,\GlobalJudge{\xasync{return}}{\Phi}}
      }
    \]
    \[
      \seqRule{AsyncCall}{
        i\text{ fresh }\\
        \seq{\Gamma}{\GlobalJudge{\update\{\invocEv(m,i)\} s}{\Phi}}
      }{ 
        \seq{\Gamma}{\GlobalJudge{\update\, !m()\xasync{;}s}{\Phi}}
      }
    \]
  \end{minipage}
  \begin{minipage}{0.49\textwidth}
  \[
    \seqRule{Open}{%
      \seq{\Gamma}{\GlobalJudge{\update  \upP{\openEvP{}{f}} s }{\Phi}} }
    {\seq{\Gamma}{\GlobalJudge{\update \, \xasync{open(}f\xasync{);}s}}{\Phi }}
  \]
  \[
    \seqRule{Close}{%
      \sequent{}{\judge{\update}{ \noEvent{} \openEvP{}{f} \noEvent{\closeEvP{}{f}}}}\\
      \sequent{}{\GlobalJudge{\update  \upP{\closeEvP{}{f}} s }{\Phi}}}
    {\sequent{}{\GlobalJudge{\update \, \xasync{close(}f\xasync{);}s}}{\Phi }}
  \]
  \[
    \seqRule{Read}{%
      \seq{\Gamma}{\judge{\update}{ \noEvent{} \openEvP{}{f} \noEvent{\closeEvP{}{f}}}}\\
      \seq{\Gamma}{\GlobalJudge{\update  \upP{\readEvP{}{f}} s }{\Phi}}}
    {\seq{\Gamma}{\GlobalJudge{\update \, \xasync{read(}f\xasync{);}s}}{\Phi }}
  \]
  \[
    \seqRule{Write}{%
      \seq{\Gamma}{\judge{\update}{ \noEvent{} \openEvP{}{f} \noEvent{\closeEvP{}{f}}}}\\
      \seq{\Gamma}{\GlobalJudge{\update \upP{\writeEvP{}{f}} s }{\Phi}}}
    {\seq{\Gamma}{\GlobalJudge{\update \, \xasync{write(}f\xasync{);}s}}{\Phi }}
  \]
\end{minipage}
\caption{Sequent rules for straight-line programs}
\label{fig:straight-line-programs}
\end{figure}

\paragraph{Synchronous Procedure Calls.}

The pattern in the conclusion of the \textsf{Call} rule below matches
the antecedent (to retrieve the contract for the called procedure
$m$), the executed statement (to ensure the next statement to be
executed is a synchronous call) and the trace specification.  The
latter splits into three parts: $\Phi$ is the specification of the
trace until the call to $m$, $\theta$ is the specification for the
part of the trace that is generated by $m$, and $\Psi$ specifies the
remaining trace.\footnote{In practice, this split shape must be
  obtained by suitable weakening rules on trace formulas. The details
  are future work.}

The rule has three premises.
The first corresponds to establishing the condition for the contract
to apply: both, the pre-trace $\theta_{a_m}$ of the procedure contract
and the prefix $\Phi$ of the specification of the currently executed
code must be valid for the current trace update $\update$.  The second
premise checks that the internal behavior $\theta_{s_m}$ specified in
the contract is contained in the given specification $\theta$, i.e.\
the contract is suitable to achieve the claimed specification. In the
last premise update $\update$ is extended with a run event to mark
that the contract was used, and the specification is strengthened by
the contract, obtained from the proof of the first two premises.
Symbolic execution continues on $s$ in the succedent, where not only
$\Psi$ needs to be established, but also $\theta_{c_m}$, because as
the caller of $m$ we are also responsible to ensure the post-trace of
the contract.
%
\[
\seqRule{Call}{%
\seq{\Gamma}{\judge{\mathcal{U}}{(\Phi \wedge \theta_{a_m})}}\qquad
\eval{\theta_{s_m}}{} \subseteq \eval{\theta}{} \qquad \text{ $i$ fresh}\\
\seq{\Gamma, \,\judge{\update \{\runEv(m,i,sy)\}}{(\Phi \wedge \theta_{a_m})\chop\,\theta_{s_m}}}{\GlobalJudge{\\ 
\qquad \update  
\{ \runEv(m,i,sy)\}\,s}{(\Phi \wedge \theta_{a_m})\chop\,\theta_{s_m}\chop\,(\Psi \wedge \theta_{c_m})}}}
{\seq{\Gamma, \judge{m}{\assumeAndContinue{\theta_{a_m}}{\theta_{s_m}}{\theta_{c_m}}} }{\GlobalJudge{\mathcal{U}\,m();s}}{\Phi\chop\,\theta\chop\,\Psi}}
\]

The remaining call rules follow the pattern established 
above with variations due to scheduling.

\paragraph{Deterministic Asynchronous Procedure Calls.}

The rule for deterministic scheduling applies when exactly one
asynchronously called procedure, here $m$, can be scheduled, according
to the first premise.  The lookup of the contract for $m$ in $\Gamma$
now becomes the second premise.  The remaining premises are analogues
to the \textsf{Call} rule, with the difference that the call
identifier $i$ is not fresh, but it matches the identifier of the
asynchronous invocation to be scheduled.

\[
\seqRule{ScheduleD}{%
 schedule(\update) = \{(m,i) \}\qquad
\seq{\Gamma}{\judge{m}{\assumeAndContinue{\theta_{a_m}}{\theta_{s_m}}{\theta_{c_m}}}}\\
\seq{\Gamma}{\judge{\update}{(\Phi \wedge \theta_{a_m})}}\qquad
\eval{\theta_{s_m}}{} \subseteq \eval{\theta}{} \\
\seq{\Gamma, \,  \judge{\update\upP{\runEv(m,i,as)}}{(\Phi \wedge \theta_{a_m}) \chop \theta_{s_m}}}{\\ \qquad \GlobalJudge{\update \upP{\runEv(m,i,as)} }{(\Phi \wedge \theta_{a_m})\chop\,\theta_{s_m}\chop\,(\Psi \wedge \theta_{c_m})}}}
{\seq{\Gamma }{\GlobalJudge{\update}}{\Phi\chop\,\theta\chop\,\Psi}}
\]

\paragraph{Non-deterministic Asynchronous Procedure Calls.}

The rule generalizes the deterministic version, by adding 
the last four premises not just once, but for each possible scheduling
decision. This ensures that symbolic execution considers all possible
traces.  We suffer from a blow-up in the size of the proof tree here,
but the use of contracts provides at least a suitable mechanism to
abstract over the scheduling decisions of the called procedures.
\[
  \seqRule{ScheduleN}{%
    \hspace*{-1cm}
(m,i) \in schedule(\update)
\begin{cases}
\seq{\Gamma}{\judge{m}{\assumeAndContinue{\theta_{a_m}}{\theta_{s_m}}{\theta_{c_m}}}}\\
\seq{\Gamma}{\judge{\update}{(\Phi \wedge \theta_{a_m})}}\qquad
\eval{\theta_{s_m}}{} \subseteq \eval{\theta}{} \\
\seq{\Gamma, \,  \judge{\update\upP{\runEv(m,i,as)}}{(\Phi \wedge \theta_{a_m}) \chop \theta_{s_m}}}{\\ \qquad \GlobalJudge{\update \upP{\runEv(m,i,as)} }{(\Phi \wedge \theta_{a_m})\chop\,\theta_{s_m}\chop\,(\Psi \wedge \theta_{c_m})}}
\end{cases}}
{\seq{\Gamma }{\GlobalJudge{\update }}{\Phi\chop\,\theta\chop\,\Psi}}
\]

We can observe a substantial degree of uniformity among the different
call rules. This is possible, because the use of events allows us to
separate \emph{scheduling} from \emph{contract application}.

\paragraph{Other Rules.}

The following inconspicuous rule reduces global to local reasoning:
When no invocation event is left to be resolved, then the local and
global judgments are equivalent and the final $\popEv$ event is
added.  We assume that the identifier and name of the procedure we are
considering are globally known as $m, \mathit{oId}$.
\[
  \seqRule{Finish}{
  \seq{\Gamma}{\judge{\update\upP{\popEvP{}{(m,\mathit{oId})}}}{\Phi}} \qquad \mathit{schedule}(\update) = \emptyset
  }{\seq{\Gamma}{\GlobalJudge{\update}{\Phi}}}
\]

\subsection{Properties of the Proof Rules}

Our rules are sound in the usual sense of sequent calculi.  As for the
compositionality, we get that all procedures behave as specified and
that files are treated correctly.

\begin{proposition}
  \label{prop:glob-sem-comp-update}
  $\eval{\update \update'}{\tau} = \bigcup_{\tau' \in \eval{\update}{\tau}} \eval{\update}{\tau} \,\chopSem\,
    \eval{\update'}{\tau \chopSem \tau'}$.
    
\end{proposition}

\begin{theorem}[Soundness]\label{thm:sound}
  Rules {\normalfont \ruleName{ScheduleD}, \ruleName{ScheduleN},
    \ruleName{Call}, \ruleName{Contract}, \ruleName{Finish}}, and the
  rules in Fig.~\ref{fig:straight-line-programs} are sound.
\end{theorem}

The following theorem states a sufficient condition for a program
without asynchronous self-calls to be correct.
  
\begin{theorem}[Global Adherence]\label{thm:global}
  Let $P$ be an always terminating program with procedures $\many{m}$.
  Let $C^{\many{m}} = \{m : C_m\}$ denote the set of all contract
  judgments and
  $C_{m'}^{\many{m}} = C^{\many{m}} \setminus \{m' : C_{m'}\}$ the set
    of contract judgments for all procedures but $m'$.  If for all
    $m\in \many{m}$ the following sequent is valid
  \[
    \seq{C^{\many{m}}_m}{m : C_m}
  \]
  then
  \begin{enumerate*}[label=(\arabic*)]
  \item All traces of $P$ are file-correct (Def.~\ref{def:file}), and
  \item $\models P$ (Def.~\ref{def:adhere}).
  \end{enumerate*}
\end{theorem}

The proofs can be found in the appendix. The only detail we point out
here is that to show program correctness, we need strong procedure
adherence, which we obtain from the proof of $\theta_{c_m}$ demanded
by the succedent of the final premise of the call rules.


\section{Case Study}\label{sec:case}
We verify the procedures described in Example~\ref{ex:files} to
illustrate the working of the calculus.  First we specify a set of
contracts for Example~\ref{ex:files}, in slightly prettified
syntax. The init block is, as discussed above, regarded as a procedure
(with a trivial contract). Obviously, this contract cannot have a
non-trivial context:
\[
  C_\xasync{main} = \assumeAndContinue{\stateFml{\mathsf{true}}}{\noEvent{}}{\stateFml{\mathsf{true}}}
\]

Regarding procedure \async{do}, we specify that it assumes that the
file stored in $f$ was not opened so far, and it closes it internally.
\[
C_\xasync{do} = \assumeAndContinue{\noEvent{}\mho~\xasync{file}~\mathbf{as}~f.~\noEvent{\openEvP{}{f}}\stateFml{\mathsf{true}}}{\noEvent{}\closeEvP{}{f}\noEvent{}}{\stateFml{\mathsf{true}}\noEvent{}}
\]

Procedure \async{closeF} specifies that it closes the file and does
not reopen it. To prove this modularly, we must specify that the file
was opened before.
\begin{align*}
  C_\xasync{closeF} = & \assume{\noEvent{}\mho~\xasync{file}~\mathbf{as}~f.~\openEvP{}{f}\noEvent{\closeEvP{}{f}}\stateFml{\mathsf{true}}} \ | \\
  & \qquad\continue{\stateFml{\mathsf{true}}\concat\closeEvP{}{f}\noEvent{\openEvP{}{f}} \ | \ \stateFml{\mathsf{true}}\noEvent{}}
\end{align*}

Finally, we specify that \async{operate} just writes to the file, but
doesn't close it, which is expected to be done by the caller.
\begin{align*}
  C_\xasync{operate} = & \assume{\noEvent{}\mho~\xasync{file}~\mathbf{as}~f.~\openEvP{}{f}\noEvent{\closeEvP{}{f}}\stateFml{\mathsf{true}}} \ | \\
  & \qquad\continue{\stateFml{\mathsf{true}}\concat\writeEvP{}{f}\noEvent{\closeEvP{}{f}} \ | \ \stateFml{\mathsf{true}}\noEvent{}\closeEvP{}{f}\noEvent{}}
\end{align*}

\paragraph{Proving \async{closeF}.}

Let $M$ be the set of all procedures. We use, for readability, the
following abbrevations:
\begin{align*}
\update &= \mathcal{V}\{\startEv(\xasync{closeF},\mathit{oId})\}\\
\phi_1 &= \judge{\update}{\openEvP{}{f}\noEvent{\closeEvP{}{f}}\stateFml{\mathsf{true}}}
\end{align*}

We apply the $\mathsf{Contract}$ rule and skolemize away all
observation quantifiers.  In the antecedent we have the assumption
that the pre-trace holds, and must show that the trace specification
is globally fulfilled.  Next, we apply the $\mathsf{Close}$ rule and
show that the file is not closed yet.  For simplicity, we map the
variable \async{f} directly to the Skolem constant from the
observation quantifier.  To show the left premise, we observe that
$\phi_1$ occurs on both sides of the sequent if we use the simple
observation $\eval{\noEvent{}}{} \Leftrightarrow \eval{\noEvent{}\stateFml{\mathsf{true}}}{}$.

\noindent\resizebox{\textwidth}{!}{
\begin{minipage}{1.1\textwidth}
\begin{prooftree}
\AxiomC{}
\UnaryInfC{$C_\mathtt{closeF}^{\many{m}}, \phi_1 \vdash
\judge{\update}{\noEvent{}\openEvP{}{f}\noEvent{\closeEvP{}{f}}}$}
\AxiomC{$\vdots\;(1)$}
\LeftLabel{$\mathsf{Close}$}
\BinaryInfC{$C_\mathtt{closeF}^{\many{m}}, \phi_1\vdash \GlobalJudge{\update \xasync{close(f); return;}}{\openEvP{}{f}\noEvent{\closeEvP{}{f}} \stateFml{\mathsf{true}}\concat\closeEvP{}{f}\noEvent{\openEvP{}{f}}}$}
\LeftLabel{$\mathsf{Contract}$}
\UnaryInfC{$ C_\mathtt{closeF}^{\many{m}} \vdash \GlobalJudge{\mathtt{closeF}}{C_\mathtt{closeF}} $}
\end{prooftree}
\end{minipage}}
\smallskip

We symbolically execute \xasync{return}, adding the corresponding
event to the update, and apply the $\mathsf{Finish}$ rule, because we
obviously cannot schedule any other procedure.  Clearly, this relies
on proof obligation generation and the concurrency model: No event
\emph{this process} can schedule is pending.

\noindent\resizebox{\textwidth}{!}{
\begin{minipage}{1.1\textwidth}
\begin{prooftree}
\AxiomC{}
\UnaryInfC{$\mathit{schedule}(\update\{\closeEvP{}{f}\}\{\retEv(\mathit{oId})\}) = \emptyset$}
\AxiomC{$\vdots\;(2)$}
\LeftLabel{$\mathsf{Finish}$}
\BinaryInfC{$C_\mathtt{closeF}^{\many{m}}, \phi_1\vdash \GlobalJudge{\update\{\closeEvP{}{f}\}\{\retEv(\mathit{oId})\}}{\openEvP{}{f}\noEvent{\closeEvP{}{f}} \stateFml{\mathsf{true}}\concat\closeEvP{}{f}\noEvent{\openEvP{}{f}}}$}
\LeftLabel{$\mathsf{Return}$}
\UnaryInfC{$C_\mathtt{closeF}^{\many{m}}, \phi_1\vdash \GlobalJudge{\update\{\closeEvP{}{f}\}\xasync{return;}}{\openEvP{}{f}\noEvent{\closeEvP{}{f}} \stateFml{\mathsf{true}}\concat\closeEvP{}{f}\noEvent{\openEvP{}{f}}}$}
\UnaryInfC{$\vdots\;(1)$}
\end{prooftree}
\end{minipage}}
\smallskip

The remaining sequent is straightforward to show: The antecedent is
\judge{\update}{\openEvP{}{f}\noEvent{\closeEvP{}{f}}\stateFml{\mathsf{true}}},
so we need to show
$
\judge{\{\closeEvP{}{f}\}\{\retEv\{\mathit{oId}\}\}}{\closeEvP{}{f}\noEvent{\openEvP{}{f}}}
$.
This reduces to
$\judge{\{\retEv\{\mathit{oId}\}\}}{\noEvent{\openEvP{}{f}}}$, which
is clearly the case.
\smallskip

\noindent\resizebox{\textwidth}{!}{
\begin{minipage}{1.15\textwidth}
\begin{prooftree}
\AxiomC{}
\UnaryInfC{$\judge{\update}{\openEvP{}{f}\noEvent{\closeEvP{}{f}}\stateFml{\mathsf{true}}} \vdash \judge{\update\{\closeEvP{}{f}\}\{\retEv(\mathit{oId})\}}{\openEvP{}{f}\noEvent{\closeEvP{}{f}} \stateFml{\mathsf{true}}\concat\closeEvP{}{f}\noEvent{\openEvP{}{f}}}$}
\UnaryInfC{$\vdots\;(2)$}
\end{prooftree}
\end{minipage}}

\paragraph{Proving \async{do}.}

Proving the contract of \async{do}, which we show in full detail in
the appendix, requires to prove an asynchronous call. At the end of
symbolic execution, the $\mathit{schedule}$ function returns a
singleton set (the call to \async{closeF}), which must be taken care
of. It is the contract of this very rule that adds the information
about a $\closeEvP{}{f}$ event that is used to prove the final
post-trace.
    
\paragraph{Proving \async{operate}.}

Proving the correctness of \async{operate} is completely analogous to
proving correctness of \async{closeF}, except that a write event
instead of a close event is added.

\paragraph{Proving the Init Block.}

The init block is trivial to prove: It does not restrict its own
behavior, thus the only failure could stem from file operations or not
fulfilling the contract of called procedures.  There are no file
operations and the only called procedure has trivial pre- and
post-traces.


\section{Liskov Principle}\label{sec:liskov}

Having introduced context-aware contracts and a calculus to verify
them, we turn our attention towards a different approach to handle
contracts: behavioral subtyping in the form of a Liskov
principle~\cite{LiskovWing94}. In its original formulation, it states
that if some property is provable for elements of a class
$\mathtt{C}$, then it must be provable for all elements of any class
$\mathtt{D}$ that is a subclass of $\mathtt{C}$.

It can, however, also be understood on a contract-level~\cite{HAGH16}
by focusing on the property to be proven, not the classes.  This means
that instead of focusing on (well-behaving of) subclasses, we  focus
on (well-behaving of) their \emph{contracts}: For state contracts the
Liskov principle states that if a procedure $m$ in class $\mathtt{C}$
has a contract $\mathit{Contr}_m$, and a procedure $n$ overrides $m$
in some subclass $\mathtt{D}$ of $\mathtt{C}$, then $n$ must also
satisfy $\mathit{Contr}_m$. Moreover, we can weaken this condition to
allow for a stronger specification of $n$: the contract
$\mathit{Contr}_n$ of $n$ must be a \emph{subcontract} of
$\mathit{Contr}_m$.  We detail the notion of a \emph{subcontract}

Thus, the Liskov principle can be expressed as an order on contracts
over procedures of the same signature, it does not require a language
involving classes, objects or even subtyping.\footnote{This
    insight was used already in \cite{HS12} to formulate a Liskov
    principle for feature-oriented programming.} We only need to
formalize the notion of a sub-contract: For \emph{state} contracts
$\mathit{Contr}_1 = (\mathsf{pre}_1, \mathsf{post}_1)$,
$\mathit{Contr}_2 = (\mathsf{pre}_2, \mathsf{post}_2)$ this is
obvious: $\mathit{Contr}_2$ is a sub-contract of $\mathit{Contr}_1$ if
these two conditions hold:
\begin{enumerate}
\item Pre-Condition $\mathsf{pre}_2$ is weaker than $\mathsf{pre}_1$:
  $\mathsf{pre}_1 \rightarrow \mathsf{pre}_2$.
\item Post-Condition $\mathsf{post}_2$ is stronger than
  $\mathsf{post}_1$: $\mathsf{post}_2 \rightarrow \mathsf{post}_1$.
\end{enumerate}

Equivalently, we say $\mathit{Contr}_1$ is \emph{more general than}
$\mathit{Contr}_2$ and write
$\mathit{Contr}_1 \succeq \mathit{Contr}_2$.  The sub-contract may
define \emph{additional} pre-conditions and a stronger, more specific
\emph{post-condition}.

We reformulate the above definition for trace contracts.  A
sub-contract may weaken the pre-condition, which is the responsibility
of the \emph{caller}.  It may give the caller more possibilites, but
restricts its own post-condition, which it can control.  The principle
is that sub-contracts \emph{weaken} the part of the contract they
cannot control, and strenghthen the part of the contract they can
control.

For trace contracts, the central issue are the formulas of the traces
$q_a$ and $q_c$ that are at the border between the assume/continue
context and the internal behavior (see Def.~\ref{def:contract}).
The subtyping principle for context-aware contracts considers $q_a$ as
part of the pre-trace
because it describes a state the procedure
\emph{cannot control}.  In contrast, $q_c$ is considered part of the
post-trace, because the specified procedure \emph{can control} it.

For simplicity, we formulate context-aware behavioral subtyping for
contracts without observation variables, whose addition is
straightforward, but technically distracting.

\begin{definition}[Behavioral Subtyping for Context-Aware Contracts]
  We define the \emph{more general than} relation $\succeq$ between
  two context-aware contracts distinguished by superscripts $1$, $2$:
  
  \begin{align}
    & \assumeAndContinueShort{\theta'_{a^1} \!\concat\! \stateFml{q_{a^1}}}{\stateFml{q_{a^1}} \!\concat\! \, \theta_{s^1}' \!\concat\! \stateFml{q_{c^1}}}{ \stateFml{q_{c^1}} \!\concat\! \theta'_{c^1}} \notag\\
    & \qquad\qquad \succeq \quad
      \assumeAndContinueShort{\theta'_{a^2} \!\concat\! \stateFml{q_{a^2}}}{\stateFml{q_{a^2}} \!\concat\! \, \theta_{s^2}' \!\concat\! \stateFml{q_{c^2}}}{\stateFml{q_{c^2}} \!\concat\!  \theta'_{c^2}} \notag\\
    &\iff\notag\\
    &\eval{\theta'_{a^1} \concat \stateFml{q_{a^1}} }{} \subseteq \eval{\theta'_{a^2} \concat \stateFml{q_{a^2}}}{} \tag{L1}\label{l1}\\
    \,\wedge\,&\eval{\theta'_{s^1} \concat \stateFml{q_{c^1}} }{} \supseteq \eval{\theta'_{s^2} \concat \stateFml{q_{c^2}} }{} \tag{L2}\label{l2}\\
    \,\wedge\,&\eval{\theta'_{c^1}}{} \subseteq \eval{\theta'_{c^2}}{} \tag{L3}\label{l3}
  \end{align}

\end{definition}

When all traces $\theta_a'$, $\theta_c'$, and $\theta_s'$ are empty,
this definition boils down to the Liskov principle for state contracts
stated above.
The first condition (\ref{l1}, for pre-trace and -condition) and the
last condition (\ref{l3}, for the post-trace) are concerned with the
context
of a procedure, which it cannot control.  Hence, we permit
weakening here. The second condition (\ref{l2}, for the internal
behavior and post-condition of the procedure) is under its control.
Thus, we permit strengthening.

In our setting without inheritance, we can use the Liskov principle to
specify a procedure with a set of contracts, and use the above
subtyping principle to order them.  First, we introduce the idea of a
maximal contract.
\begin{definition}[Maximal Procedure Contract]
  Given a finite set $\mathcal{N}$ of procedure contracts for a procedure
  ordered by $\succeq$, let $Max(\mathcal{N})$ be the set of maximal
  elements in $\mathcal{N}$.
\end{definition}

We permit a procedure to be specified with multiple contracts, for
different situations according to different usages.  Using
\emph{maximal} procedure contracts, we need only a subset of them to
be proven.
The following rule is straightforward: Given an invocation event on
$m$, it computes all maximal (most general) contracts and applies all
of them.
%
%
\[
\seqRule{actOrder}{%
  \hspace*{-1cm} schedule(\update) = \mathcal{P} \\
  \hspace*{-1cm}
(m,i) \in Max(\mathcal{P})
\begin{cases}
\seq{\Gamma}{\judge{m}{\assumeAndContinue{\theta_{a_m}}{\theta_{s_m}}{\theta_{c_m}}}}\\
\seq{\Gamma}{\judge{\update}{(\Phi \wedge \theta_{a_m})}}\qquad
\eval{\theta_{s_m}}{} \subseteq \eval{\theta}{} \\
\seq{\Gamma, \,  \judge{\update\upP{\runEv(m,i,as)}}{(\Phi \wedge \theta_{a_m}) \chop \theta_{s_m}}}{\\ \qquad \GlobalJudge{\update \upP{\runEv(m,i,as)} }{(\Phi \wedge \theta_{a_m})\chop\,\theta_{s_m}\chop\,(\Psi \wedge \theta_{c_m})}}
\end{cases}}
{\seq{\Gamma }{\GlobalJudge{\update }}{\Phi\chop\,\theta\chop\,\Psi}}
\]

\begin{remark}
  We proposed behavioral subtyping as a technique to reduce the number
  of different call sequences of asynchronous procedures that need to
  be considered during verification. Techniques to combat
  combinatorial explosion of instruction sequences are standard in
  model checking~\cite{CGMP99}, but they are at the level of
  \emph{code} not at the level of specifications (i.e.,
  contracts). This holds even for a partial-order reduction technique
  that was adapted to asynchronous procedures~\cite{ABGIS19}.
\end{remark}


\section{Conclusion}\label{sec:conc}

We extended state contracts and trace contracts to \emph{context-aware
  trace contracts} (CATs). This permits to specify the behavioral
context in which a procedure is executed, i.e., not merely the static
pre-condition, but the actions and states reached before its execution
begins and also those after it ends. Such a specification of the
\emph{call context} as part of a procedure contract is essential to
specify the global behavior of concurrent programs is a succinct
manner. We instantiated the CAT methodology to a language with
asynchronous calls, where the context is of uttermost importance, and
gave a proof-of-concept using a file handling scenario.

To combat combinatorial explosion in verification proofs, we stated a
Liskov principle for CATs that has the potential to reduce the effort
dramatically.

We hope our work will enable new specification patterns to overcome
the long-standing specification challenge in deductive verification
for trace properties and concurrent programs.

\paragraph{Future Work.} 

In future work, we plan to investigate richer concurrency models, in
particular full Active Objects with suspension, futures and multiple
objects \cite{ActiveObjects17}.  One question that we did not
investigate here, is the relation of CATs to object and system
invariants.

Our observation quantifiers permit to connect programs and traces in a
reasonably abstract, yet fine-grained and flexible manner. However,
their proof theory is an open question: In rule \textsf{Contract} in
Sect.~\ref{sec:proof-rules} we approximated the semantics of
observation quantifiers in a fairly crude manner by
Skolemization. This makes it difficult to relate different
observations to each other and draw conclusions from them. The
axiomatization of observation quantifiers and their comparison to
existing logics with observation constructs should be studied in its
own right.

Obviously, the CAT framework must be implemented so that larger case
studies can be performed. This involves to complete the existing rule
set with rules for update simplification in local judgments.




\subsubsection{Acknowledgements}

This work was partially supported by the Research Council of Norway
via the \texttt{SIRIUS} Centre (237898) and the \texttt{PeTWIN}
project (294600), as well as the Hessian LOEWE initiative within the
Software-Factory 4.0 project.

We profited enormously from the detailed and constructive remarks of
the reviewers.

\bibliographystyle{abbrv}
\bibliography{ref,reiner}

\begin{thebibliography}{10}

\bibitem{DBLP:series/lncs/10001}
W.~Ahrendt, B.~Beckert, R.~Bubel, R.~H{\"{a}}hnle, P.~H. Schmitt, and
  M.~Ulbrich, editors.
\newblock {\em Deductive Software Verification - The KeY Book - From Theory to
  Practice}, volume 10001 of {\em Lecture Notes in Computer Science}.
\newblock Springer, 2016.

\bibitem{ABGIS19}
E.~Albert, M.~G. de~la Banda, M.~G{\'{o}}mez{-}Zamalloa, M.~Isabel, and P.~J.
  Stuckey.
\newblock Optimal context-sensitive dynamic partial order reduction with
  observers.
\newblock In D.~Zhang and A.~M{\o}ller, editors, {\em Proc.\ 28th {ACM}
  {SIGSOFT} Intl.\ Symp.\ on Software Testing and Analysis, {ISSTA}}, pages
  352--362. {ACM}, 2019.

\bibitem{DBLP:conf/oopsla/AldrichSSS09}
J.~Aldrich, J.~Sunshine, D.~Saini, and Z.~Sparks.
\newblock Typestate-oriented programming.
\newblock In {\em {OOPSLA} Companion}, pages 1015--1022. {ACM}, 2009.

\bibitem{BBBB12}
C.~Baumann, B.~Beckert, H.~Blasum, and T.~Bormer.
\newblock Lessons learned from microkernel verification -- specification is the
  new bottleneck.
\newblock In F.~Cassez, R.~Huuck, G.~Klein, and B.~Schlich, editors, {\em
  Proc.\ 7th Conference on Systems Software Verification}, volume 102 of {\em
  EPTCS}, pages 18--32, 2012.

\bibitem{bern}
B.~Beckert and D.~Bruns.
\newblock Dynamic logic with trace semantics.
\newblock In M.~P. Bonacina, editor, {\em Automated Deduction - {CADE} 2013.
  Proceedings}, volume 7898 of {\em Lecture Notes in Computer Science}, pages
  315--329. Springer, 2013.

\bibitem{BubelCHN15}
R.~Bubel, C.~C. Din, R.~H{\"a}hnle, and K.~Nakata.
\newblock A dynamic logic with traces and coinduction.
\newblock In H.~D. Nivelle, editor, {\em Intl.\ Conf.\ on Automated Reasoning
  with Analytic Tableaux and Related Methods, Wroclaw, Poland}, volume 9323 of
  {\em LNCS}, pages 303--318. Springer, 2015.

\bibitem{DBLP:journals/corr/abs-2211-09487}
R.~Bubel, D.~Gurov, R.~H{\"a}hnle, and M.~Scaletta.
\newblock Trace-based deductive verification.
\newblock In R.~Piskac and A.~Voronkov, editors, {\em Proc.\ 20th Intl.\ Conf.\
  on Logic for Programming, Artificial Intelligence and Reasoning (LPAR),
  Manizales Colombia}, EPiC Series in Computing. EasyChair, 2023.

\bibitem{CGMP99}
E.~M. Clarke, O.~Grumberg, M.~Minea, and D.~A. Peled.
\newblock State space reduction using partial order techniques.
\newblock {\em Int. J. Software Tools for Technology Transfer}, 2(3):279--287,
  1999.

\bibitem{ActiveObjects17}
F.~de~Boer, C.~C. Din, K.~Fernandez-Reyes, R.~H{\"a}hnle, L.~Henrio, E.~B.
  Johnsen, E.~Khamespanah, J.~Rochas, V.~Serbanescu, M.~Sirjani, and A.~M.
  Yang.
\newblock A survey of active object languages.
\newblock {\em ACM Computing Surveys}, 50(5):76:1--76:39, Oct. 2017.
\newblock Article 76.

\bibitem{TimSort17}
S.~De~Gouw, F.~S. De~Boer, R.~Bubel, R.~H{\"a}hnle, J.~Rot, and
  D.~Steinh{\"o}fel.
\newblock Verifying {OpenJDK's} sort method for generic collections.
\newblock {\em J. Automated Reasoning}, 62(1), 2019.

\bibitem{DBLP:conf/ecoop/DeLineF04}
R.~DeLine and M.~F{\"{a}}hndrich.
\newblock Typestates for objects.
\newblock In {\em {ECOOP}}, volume 3086 of {\em Lecture Notes in Computer
  Science}, pages 465--490. Springer, 2004.

\bibitem{DinBH15}
C.~C. Din, R.~Bubel, and R.~H{\"a}hnle.
\newblock {KeY-ABS}: A deductive verification tool for the concurrent modelling
  language {ABS}.
\newblock In A.~Felty and A.~Middeldorp, editors, {\em Proc.\ 25th Intl.\
  Conf.\ on Automated Deduction (CADE), Berlin, Germany}, volume 9195 of {\em
  LNCS}, pages 517--526. Springer, 2015.

\bibitem{DBLP:journals/corr/abs-2202-12195}
C.~C. Din, R.~H{\"{a}}hnle, L.~Henrio, E.~B. Johnsen, V.~K.~I. Pun, and
  S.~L.~T. Tarifa.
\newblock {LAGC} semantics of concurrent programming languages.
\newblock {\em CoRR}, abs/2202.12195, 2022.

\bibitem{DHJPT17}
C.~C. Din, R.~H{\"a}hnle, E.~B. Johnsen, V.~K.~I. Pun, and S.~L. Tapia~Tarifa.
\newblock Locally abstract, globally concrete semantics of concurrent
  programming languages.
\newblock In C.~Nalon and R.~Schmidt, editors, {\em Proc.\ 26th Intl.\ Conf.\
  on Automated Reasoning with Tableaux and Related Methods}, volume 10501 of
  {\em LNCS}, pages 22--43. Springer, Sept. 2017.

\bibitem{DinOwe15}
C.~C. Din and O.~Owe.
\newblock Compositional reasoning about active objects with shared futures.
\newblock {\em Formal Aspects of Computing}, 27(3):551--572, 2015.

\bibitem{Larch93}
J.~V. Guttag, J.~J. Horning, S.~J. Garland, K.~D. Jones, A.~Modet, and J.~M.
  Wing.
\newblock {\em {Larch}: Languages and Tools for Formal Specification}.
\newblock Springer-Verlag, New York, 1993.

\bibitem{HaehnleHuisman19}
R.~H{\"a}hnle and M.~Huisman.
\newblock Deductive verification: from pen-and-paper proofs to industrial
  tools.
\newblock In B.~Steffen and G.~Woeginger, editors, {\em Computing and Software
  Science: State of the Art and Perspectives}, volume 10000 of {\em LNCS},
  pages 345--373. Springer, Cham, Switzerland, 2019.

\bibitem{HS12}
R.~H{\"a}hnle and I.~Schaefer.
\newblock A {L}iskov principle for delta-oriented programming.
\newblock In T.~Margaria and B.~Steffen, editors, {\em Leveraging Applications
  of Formal Methods, Verification and Validation. Technologies for Mastering
  Change --- 5th International Symposium, ISoLA 2012, Heraklion, Crete,
  Greece}, volume 7609 of {\em LNCS}, pages 32--46. Springer, Oct. 2012.

\bibitem{HalpernShoham91}
J.~Y. Halpern and Y.~Shoham.
\newblock A propositional modal logic of time intervals.
\newblock {\em Journal of the ACM}, 38(4):935--962, 1991.

\bibitem{HarelKP80}
D.~Harel, D.~Kozen, and R.~Parikh.
\newblock Process logic: Expressiveness, decidability, completeness.
\newblock In {\em 21st Annual Symposium on Foundations of Computer Science,
  Syracuse, New York, USA, 13-15 October 1980}, pages 129--142. {IEEE} Computer
  Society, 1980.

\bibitem{Honda08}
K.~Honda, N.~Yoshida, and M.~Carbone.
\newblock Multiparty asynchronous session types.
\newblock In {\em Proceedings of the 35th {ACM} {SIGPLAN-SIGACT} Symposium on
  Principles of Programming Languages, {POPL} 2008}, pages 273--284, 2008.

\bibitem{HAGH16}
M.~Huisman, W.~Ahrendt, D.~Grahl, and M.~Hentschel.
\newblock Formal specification with the {Java Modeling Language}.
\newblock In W.~Ahrendt, B.~Beckert, R.~Bubel, R.~H\"{a}hnle, P.~Schmitt, and
  M.~Ulbrich, editors, {\em Deductive Software Verification---The {KeY} Book:
  From Theory to Practice}, volume 10001 of {\em LNCS}, chapter~7, pages
  193--241. Springer, 2016.

\bibitem{ABSFMCO10}
E.~B. Johnsen, R.~H{\"a}hnle, J.~Sch{\"a}fer, R.~Schlatte, and M.~Steffen.
\newblock {ABS}: A core language for abstract behavioral specification.
\newblock In B.~K. Aichernig, F.~de~Boer, and M.~M. Bonsangue, editors, {\em
  Proc.\ 9th International Symposium on Formal Methods for Components and
  Objects ({FMCO} 2010)}, volume 6957 of {\em Lecture Notes in Computer
  Science}, pages 142--164. Springer-Verlag, 2011.

\bibitem{Jones81}
C.~B. Jones.
\newblock {\em Developing methods for computer programs including a notion of
  interference}.
\newblock PhD thesis, University of Oxford, {UK}, 1981.

\bibitem{Jones95b}
C.~B. Jones.
\newblock Granularity and the development of concurrent programs.
\newblock In S.~D. Brookes, M.~G. Main, A.~Melton, and M.~W. Mislove, editors,
  {\em 11th Annual Conf.\ on Mathematical Foundations of Programming Semantics,
  {MFPS}, New Orleans, LA, USA}, volume~1 of {\em ENTCS}, pages 302--306.
  Elsevier, 1995.

\bibitem{DBLP:conf/tableaux/Kamburjan19}
E.~Kamburjan.
\newblock Behavioral program logic.
\newblock In {\em {TABLEAUX}}, volume 11714 of {\em Lecture Notes in Computer
  Science}, pages 391--408. Springer, 2019.

\bibitem{DBLP:phd/dnb/Kamburjan20}
E.~Kamburjan.
\newblock {\em Modular Verification of a Modular Specification: Behavioral
  Types as Program Logics}.
\newblock PhD thesis, Darmstadt University of Technology, Germany, 2020.

\bibitem{DBLP:conf/ifm/KamburjanC18}
E.~Kamburjan and T.~Chen.
\newblock Stateful behavioral types for active objects.
\newblock In {\em {IFM}}, volume 11023 of {\em Lecture Notes in Computer
  Science}, pages 214--235. Springer, 2018.

\bibitem{DBLP:conf/icfem/KamburjanDC16}
E.~Kamburjan, C.~C. Din, and T.~Chen.
\newblock Session-based compositional analysis for actor-based languages using
  futures.
\newblock In {\em {ICFEM}}, volume 10009 of {\em Lecture Notes in Computer
  Science}, pages 296--312, 2016.

\bibitem{DBLP:series/lncs/KamburjanDHJ20}
E.~Kamburjan, C.~C. Din, R.~H{\"{a}}hnle, and E.~B. Johnsen.
\newblock Behavioral contracts for cooperative scheduling.
\newblock In {\em 20 Years of KeY}, volume 12345 of {\em LNCS}. Springer, 2020.

\bibitem{crowbar}
E.~Kamburjan, M.~Scaletta, and N.~Rollshausen.
\newblock \textit{Deductive Verification of Active Objects with Crowbar}.
\newblock {\em Sci. Comput. Program.}, 226, 2023.

\bibitem{Kassios11}
I.~T. Kassios.
\newblock The dynamic frames theory.
\newblock {\em Formal of Aspects Computing}, 23(3):267--288, 2011.

\bibitem{JML-Ref-Manual}
G.~T. Leavens, E.~Poll, C.~Clifton, Y.~Cheon, C.~Ruby, D.~Cok, P.~M\"{u}ller,
  J.~Kiniry, P.~Chalin, D.~M. Zimmerman, and W.~Dietl.
\newblock {\em JML Reference Manual}, May 2013.
\newblock Draft revision 2344.

\bibitem{LiskovWing94}
B.~Liskov and J.~M. Wing.
\newblock A behavioral notion of subtyping.
\newblock {\em {ACM} Trans. Program. Lang. Syst.}, 16(6):1811--1841, 1994.

\bibitem{Meyer92}
B.~Meyer.
\newblock Applying ``design by contract''.
\newblock {\em IEEE Computer}, 25(10):40--51, Oct. 1992.

\bibitem{DBLP:journals/corr/abs-2209-05136}
J.~Mota, M.~Giunti, and A.~Ravara.
\newblock On using verifast, vercors, plural, and key to check object usage.
\newblock {\em CoRR}, abs/2209.05136, 2022.

\bibitem{MSS16}
P.~M{\"{u}}ller, M.~Schwerhoff, and A.~J. Summers.
\newblock {Viper:} {A} verification infrastructure for permission-based
  reasoning.
\newblock In B.~Jobstmann and K.~R.~M. Leino, editors, {\em Verification, Model
  Checking, and Abstract Interpretation, 17th Intl.\ Conf., {VMCAI}, St.
  Petersburg, FL, USA}, volume 9583 of {\em LNCS}, pages 41--62. Springer,
  2016.

\bibitem{NakataUustalu15}
K.~Nakata and T.~Uustalu.
\newblock A {H}oare logic for the coinductive trace-based big-step semantics of
  {While}.
\newblock {\em Logical Methods in Computer Science}, 11(1):1--32, 2015.

\bibitem{DBLP:conf/concur/OHearn04}
P.~W. O'Hearn.
\newblock Resources, concurrency and local reasoning.
\newblock In {\em {CONCUR}}, volume 3170 of {\em Lecture Notes in Computer
  Science}, pages 49--67. Springer, 2004.

\bibitem{Pnueli77}
A.~Pnueli.
\newblock The temporal logic of programs.
\newblock In {\em 18th Annual Symposium on Foundations of Computer Science,
  Providence, Rhode Island, USA}, pages 46--57. {IEEE} Computer Society, 1977.

\bibitem{DBLP:conf/lics/Reynolds02}
J.~C. Reynolds.
\newblock Separation logic: {A} logic for shared mutable data structures.
\newblock In {\em {LICS}}, pages 55--74. {IEEE} Computer Society, 2002.

\bibitem{Wolper83}
P.~Wolper.
\newblock Temporal logic can be more expressive.
\newblock {\em Information and Control}, 56:72--99, 1983.

\end{thebibliography}

\appendix

\section{Appendix}


\subsection{Proof Theorem~\ref{thm:sound}}
To make the proofs more readable we overload some notations in the semantics 
of programs and updates as follows:
\begin{align*}
  \globEval{s}{T} = \bigcup_{\tau\in T}\globEval{s}{\tau} \qquad 
  \locEval{s}{T} = \bigcup_{\tau\in T}\locEval{s}{\tau} \qquad
  \eval{\update}{T} = \bigcup_{\tau\in T}\eval{\update}{\tau}
\end{align*}
We remind that the chop operator is also overloaded in Def.~\ref{def:judgement-global-update}.
We require also two simple propositions.

\begin{proposition}
  \label{prop:skolem-to-obs}
  Let $\Phi$ be a trace formula, $\many{y}$ be logical 
  variables and $\many{c}$ be constants fresh in $\Phi$.
$$
\eval{\Phi[\overline{y} \backslash \overline{c}]}{} 
\subseteq
\eval{\mho_{\overline{x},\overline{y}}.\Phi}{}
$$  
\end{proposition}
\begin{proof}
  \begin{align*}
    \eval{\Phi[\overline{y} \backslash \overline{c}]}{\rho,o} &=
    \{\singleton{\sigma [\many{x} \mapsto \many{c}]} \chopSem \tau ~|~
   \tau \in \evalPhiNoArg{\Phi}_{\rho, o[\many{y} \mapsto (\many{x},\sigma[\many{x} \mapsto \many{c}])]} \wedge \sigma \in \Sigma\}\\
   &\subseteq\{\singleton{\sigma} \chopSem \tau ~|~
  \tau \in \evalPhiNoArg{\Phi}_{\rho, o[\many{y} \mapsto (\many{x},\sigma)]} \wedge \sigma \in \Sigma\}
   = \eval{\mho_{\overline{x},\overline{y}}.\Phi}{\rho,o}
  \end{align*}
\end{proof}

\begin{proposition}
  \label{prop:skolem-trace-op}
  Let $\Phi_1,\Phi_2$ be trace formulas, $\many{y}$ be logical 
  variables and $\many{c}$ be constants fresh in $\Phi_1$ and $\Phi_2$.
  If $\many{y}$ do not occur in $\Phi_1$ then 
  \begin{align*}
    (\Phi_1 \chop \Phi_2)[\overline{y} \backslash \overline{c}]
    &= \Phi_1 \chop (\Phi_2[\overline{y} \backslash \overline{c}])
    \\
    (\Phi_1 \concat \Phi_2)[\overline{y} \backslash \overline{c}]
    &= \Phi_1 \concat (\Phi_2[\overline{y} \backslash \overline{c}])
  \end{align*}
\end{proposition}
\begin{proof}
Follows directly from the definitions.
\end{proof}


The following proposition shows that weak adherence is equivalent to a condition on internal trace and pre-condition
\begin{proposition}\label{prop:connect}
Let 
$$\theta_\mathit{pre} = \left(\theta'_{a_m} \concat \stateFml{q_{a_m}}\right)\left[\many{y_1} \setminus \many{c^y_1}\right]\enspace ,$$
$$\theta_\mathit{post} = \left(\theta'_{a_m} \concat \stateFml{q_{a_m}}\concat \theta_{s_m}' \concat \stateFml{q_{c_m}}\right)\left[\many{y_1} \setminus \many{c^y_1}\right]\left[\many{y_2} \setminus \many{c^y_2}\right]$$
We have
\begin{align*}
&\forall \tau \in T \chopSem \globEval{\xasync{m()}}{T}.~
\tau, oId \models \widehat{C_\mathtt{m}}\\
\iff
&\tau = \tau'\chopSem \tau''\chopSem \tau''' \wedge\\
& \tau' \in \evalPhiNoArg{\theta_\mathit{pre}\concat\mathsf{start}(m,oId)}_{\emptyset,\emptyset} \wedge \\
&\tau'' \in 
\globEval{\mathsf{inline}(m)}{\evalPhiNoArg{\mathsf{start}(m,oId)}_{\emptyset,\emptyset}} \wedge\\
&\globEval{\mathsf{inline}(m)}{\tau'\chopSem\evalPhiNoArg{\mathsf{start}(m,oId)}_{\emptyset,\emptyset}} \subseteq 
\evalPhiNoArg{\theta_\mathit{post}}
\end{align*}
\end{proposition}
\begin{proof}
We directly unroll the definition and assume a fixed identifier $oId$, which is the interpretation of the corresponding skolem symbol.
\begin{align*}
&\forall \tau \in T \chopSem \globEval{\xasync{m()}}{T}.~
\tau, oId \models \widehat{C_\mathtt{m}}\\
\iff&
    \tau \in \evalPhiNoArg{
    \theta'_{a_m} \concat \mho_{\many{x_1,y_1}}.\bigl(\stateFml{q_{a_m}} \concat \mathsf{start}(m,oId) \concat
    \stateFml{q_{a_m}}\concat \theta'_{s_m} \concat
     \mho_{\many{x_2,y_2}}.\left(\stateFml{q_{c_m}} \concat \mathsf{pop}(oId) \concat
    \stateFml{q_{c_m}}\concat \noEvent{} \right)\bigr)
}_{\emptyset,\emptyset}
\end{align*}

Next, we can split the trace formula into the part of the prefix, the procedure execution and the suffix. Note that the suffix is trivial.
\begin{align*}
&
    \tau \in \evalPhiNoArg{
    \theta'_{a_m} \concat \mho_{\many{x_1,y_1}}.\bigl(\stateFml{q_{a_m}} \concat \mathsf{start}(m,oId) \concat
    \stateFml{q_{a_m}}\concat \theta'_{s_m} \concat
     \mho_{\many{x_2,y_2}}.\left(\stateFml{q_{c_m}} \concat \mathsf{pop}(oId) \concat
    \stateFml{q_{c_m}}\concat \noEvent{} \right)\bigr)
}_{\emptyset,\emptyset}\\
\Leftarrow&
\tau = \tau'\chopSem \tau''\chopSem \tau''' \wedge\\
&\tau' \in 
    \evalPhiNoArg{\theta'_{a_m} \concat \mho_{\many{x_1,y_1}}.\stateFml{q_{a_m}} \concat \mathsf{start}(m,oId)}_{\emptyset,\emptyset}\wedge \\
&\tau'' \in 
\evalPhiNoArg{
    \mho_{\many{x_1,y_1}}.\stateFml{q_{a_m}}\concat \theta'_{s_m} \concat
     \mho_{\many{x_2,y_2}}.\left(\stateFml{q_{c_m}} \concat \mathsf{pop}(oId) \concat
    \stateFml{q_{c_m}}\right)}_{\emptyset,\emptyset}\\
    & \tau''' \in \evalPhiNoArg{\noEvent{}}_{\emptyset,\emptyset}
\end{align*}

We can now reduce, using Prop.~\ref{prop:skolem-trace-op} and Prop.~\ref{prop:skolem-to-obs}, to the following:
\begin{align*}
&\tau = \tau'\chopSem \tau''\chopSem \tau''' \wedge\\
&\tau' \in 
    \evalPhiNoArg{\theta'_{a_m} \concat \mho_{\many{x_1,y_1}}.\stateFml{q_{a_m}} \concat \mathsf{start}(m,oId)}_{\emptyset,\emptyset}\wedge \\
&\tau'' \in 
\evalPhiNoArg{
    \mathsf{start}(m,oId)\concat \mho_{\many{x_1,y_1}}.\stateFml{q_{a_m}}\concat \theta'_{s_m} \concat
     \mho_{\many{x_2,y_2}}.\left(\stateFml{q_{c_m}} \concat \mathsf{pop}(oId) \concat
    \stateFml{q_{c_m}}\right)}_{\emptyset,\emptyset}\\
    & \tau''' \in \evalPhiNoArg{\noEvent{}}_{\emptyset,\emptyset}\\
\iff& 
\tau = \tau'\chopSem \tau''\chopSem \tau''' \wedge\\
& \tau' \in \evalPhiNoArg{\theta_\mathit{pre}\concat\mathsf{start}(m,oId)}_{\emptyset,\emptyset} \wedge \\
&\tau'' \in 
\evalPhiNoArg{\mathsf{start}(m,oId)\concat
\theta'_\mathit{s_m}
\concat \mathsf{pop}(oId) \concat
    \stateFml{q_{c_m}}}_{\emptyset,\emptyset}\\
    & \tau''' \in \evalPhiNoArg{\noEvent{}}_{\emptyset,\emptyset}
\end{align*}

Finally, we observe that we can collapse the condition on the internal trace and remove the suffix.
\begin{align*}
&\tau = \tau'\chopSem \tau''\chopSem \tau''' \wedge\\
& \tau' \in \evalPhiNoArg{\theta_\mathit{pre}\concat\mathsf{start}(m,oId)}_{\emptyset,\emptyset} \wedge \\
&\tau'' \in 
\evalPhiNoArg{
\mathsf{start}(m,oId)\concat\theta'_\mathit{s_m}
\concat \mathsf{pop}(oId) \concat
    \stateFml{q_{c_m}}}_{\emptyset,\emptyset}\\
    & \tau''' \in \evalPhiNoArg{\noEvent{}}_{\emptyset,\emptyset}\\
    \Leftarrow
&\tau = \tau'\chopSem \tau''\chopSem \tau''' \wedge\\
& \tau' \in \evalPhiNoArg{\theta_\mathit{pre}\concat\mathsf{start}(m,oId)}_{\emptyset,\emptyset} \wedge \\
&\tau'' \in 
\globEval{\mathsf{inline}(m)}{\evalPhiNoArg{\mathsf{start}(m,oId)}_{\emptyset,\emptyset}} \wedge\\
&\globEval{\mathsf{inline}(m)}{\tau'\chopSem\evalPhiNoArg{\mathsf{start}(m,oId)}_{\emptyset,\emptyset}} \subseteq 
\evalPhiNoArg{\theta_\mathit{post}}
\end{align*}
\end{proof}

We now turn to the rules themselves.

\paragraph{Rule $\mathsf{Contract}$}
We prove the soundness of the rule.
Let $P$ be a program.
Let $\sigma$ be a state and $\mathcal{V},\overline{c_1^y},\overline{c_2^y}$, and $i$ 
fresh \emph{havoc} update, constants and call identifier, respectively. 
Let $$C_m= \assumeAndContinue{\theta'_{a_m} \concat \mho_{\many{x_1,y_1}}.\stateFml{q_{a_m}}}{\stateFml{q_{a_m}} \concat \, \theta_{s_m}' \concat \mho_{\many{x_2,y_2}}.\stateFml{q_{c_m}}}{ \stateFml{q_{c_m}}\concat \, \theta'_{c_m}}$$
the contract of procedure $m$.
Also, let $C$ be all the contract judgments for $P$.
With 
$$\theta_\mathit{pre} = \left(\theta'_{a_m} \concat \stateFml{q_{a_m}}\right)\left[\many{y_1} \setminus \many{c^y_1}\right]\enspace ,$$
$$\theta_\mathit{post} = \left(\theta'_{a_m} \concat \stateFml{q_{a_m}}\concat \theta_{s_m}' \concat \stateFml{q_{c_m}}\right)\left[\many{y_1} \setminus \many{c^y_1}\right]\left[\many{y_2} \setminus \many{c^y_2}\right]$$
and assuming that 
$\theta'_{a_m}$ does not contain $\many{y_1}$
and $\theta'_{s_m}$ does not contain $\many{y_2}$
we have:

$$\begin{array}{lll}
  &\evalPhiNoArg{\theta_\mathit{pre}} &= \evalPhiNoArgBig{\left(\theta'_{a_m} \concat 
  \stateFml{q_{a_m}}\right)\left[\many{y_1} \setminus \many{c^y_1}\right]}{}\\
  \mbox{\footnotesize\{Prop.~\ref{prop:skolem-trace-op}\}}
  &&=\evalPhiNoArgBig{\theta'_{a_m} \concat 
  \stateFml{q_{a_m}}\left[\many{y_1} \setminus \many{c^y_1}\right]}\\
  \mbox{\footnotesize\{Prop.~\ref{prop:skolem-to-obs}\}}
  &&\subseteq\evalPhiNoArg{\theta'_{a_m} \concat \mho_{\many{x_1},\many{y_1}}\stateFml{q_{a_m}}}
\end{array}$$

$$\begin{array}{lll}
  &\evalPhiNoArg{\theta_\mathit{post}} &= \evalPhiNoArgBig{\left(\theta'_{a_m} \concat \stateFml{q_{a_m}}\concat \theta_{s_m}' 
  \concat \stateFml{q_{c_m}}\right)\left[\many{y_1} \setminus \many{c^y_1}\right]\left[\many{y_2} \setminus \many{c^y_2}\right]}\\
  \mbox{\footnotesize\{Prop.~\ref{prop:skolem-trace-op}\}}
  &&=\evalPhiNoArgBig{\theta'_{a_m} \concat 
  \left(\stateFml{q_{a_m}} 
  \concat \theta_{s_m}' \concat \stateFml{q_{c_m}}\left[\many{y_2} \setminus \many{c^y_2}\right]\right)
  \left[\many{y_1} \setminus \many{c^y_1}\right]}\\
  \mbox{\footnotesize\{Prop.~\ref{prop:skolem-to-obs}\}}&&
  \subseteq \evalPhiNoArgBig{\theta'_{a_m} \concat \mho_{\many{x_1},\many{y_1}}
  (\stateFml{q_{a_m}}\concat \theta_{s_m}'\concat \mho_{\many{x_2},\many{y_2}}\stateFml{q_{c_m}})}
\end{array}$$

Therefore, with $T = \eval{\mathcal{V}}{\singleton{\sigma}}$:
$$\begin{array}{lll}
&&
(\models C \wedge \sigma\models \judge{\mathcal{V}}{\theta_{pre}}) 
\\
&&\quad \Rightarrow \sigma\models \GlobalJudge{\mathcal{V}\upP{\startEv(m,i)}\text{inline($m$)}}{\theta_{post}}\\
\mbox{\footnotesize\{Def.~\ref{def:local-update-judgement},~\ref{def:judgement-global-update}\}}
&\Rightarrow
&(\models C \wedge \eval{\mathcal{V}}{\singleton{\sigma}} \subseteq \eval{\theta_{pre}}{}) 
\\
&&\quad \Rightarrow 
\eval{\mathcal{V}\upP{\startEv(m,i)}}{\singleton{\sigma}} 
\chopSem 
\globEval{\text{inline($m$)}}{\eval{\mathcal{V}\upP{\startEv(m,i)}}{\singleton{\sigma}}}
\subseteq \eval{\theta_{post}}{}\\
\mbox{\footnotesize\{Prop.~\ref{prop:glob-sem-comp-update}\}}
&\Rightarrow&
(\models C \wedge T \subseteq \eval{\theta_{pre}}{}) 
\\
&&\quad \Rightarrow 
T \chopSem \eval{\upP{\startEv(m,i)}}{T} 
\\ && \qquad \chopSem \globEval{\text{inline($m$)}}{
T \chopSem \eval{\upP{\startEv(m,i)}}{T}}
\subseteq \eval{\theta_{post}}{}\\
\mbox{\footnotesize\{Fig.~\ref{fig:stmt-update-eval}\}}
&\Rightarrow&
(\models C \wedge T \subseteq \eval{\theta_{pre}}{}) 
\\
&&\quad \Rightarrow 
\forall \tau \in T.(
\tau \chopSem \callEvP{last(\tau)}{m,i}\chopSem\pushEvP{last(\tau)}{(m,i)}
\\ && \qquad \chopSem \globEval{\text{inline($m$)}}{
\tau \chopSem \callEvP{last(\tau)}{m,i}\chopSem\pushEvP{last(\tau)}{(m,i)}}
\subseteq \eval{\theta_{post}}{})\\
\mbox{\footnotesize\{Def.~$\globEval{\xasync{m()}}{T}$\}}
&\Rightarrow&
(\models C \wedge T \subseteq \eval{\theta_{pre}}{}) 
 \Rightarrow 
T \chopSem \globEval{\xasync{m()}}{T}
\subseteq \eval{\theta_{post}}{})\\
\mbox{\footnotesize\{Prop.~\ref{prop:connect}\}}
&\Rightarrow&  
\sigma \models \widehat{C_\mathtt{m}}
\end{array}$$

\paragraph{Rule $\mathsf{Assign}$}
We prove soundness and reversibility of the rule.
Let $\Gamma$ be a set of judgments and let $\sigma$ be a state with $\sigma \models \Gamma$. We have, with $T=\eval{\update}{\singleton{\sigma}}$:
  
\begin{align*}
   &\phantom{~\Leftrightarrow}  
  \sigma \models \GlobalJudge{\update\upP{v:=e}s}{\Phi} & & \\
  &\Leftrightarrow~
  \eval{\update\upP{v:=e}}{\singleton{\sigma}} \chopSem \globEval{s}{\eval{\update\upP{v:=e}}{\singleton{\sigma}}} \subseteq \evalPhiNoArg{\Phi} &
  &\mbox{\footnotesize\{Def.~\ref{def:judgement-global-update}\}} \\
  &\Leftrightarrow~
  T \,\chopSem\, \eval{\upP{v:=e}}{T} \,\chopSem\, \globEval{s}{T \chopSem \eval{\upP{v:=e}}{T}} \subseteq \evalPhiNoArg{\Phi} &
  &\mbox{\footnotesize\{Prop.~\ref{prop:glob-sem-comp-update}\}} \\
  &\Leftrightarrow~
  T \,\chopSem\, \locEval{v=e}{T} \,\chopSem\, \globEval{s}{T \chopSem \locEval{v=e}{T}} \subseteq \evalPhiNoArg{\Phi} &
  &\mbox{\footnotesize\{Def. $\valP{}{\upP{v:=e}}$\}} \\
  &\Leftrightarrow~
  T \,\chopSem\, \globEval{v=e;s}{T} \subseteq \evalPhiNoArg{\Phi} &
  &\mbox{\footnotesize\{Prop.~\ref{prop:semantics-composition}\}} \\
  &\Leftrightarrow~
  \sigma \models \GlobalJudge{\update v=e;s}{\Phi} &
  &\mbox{\footnotesize\{Def.~\ref{def:judgement-global-update}\}}
\end{align*}
We, therefore, have that the sequent
$\sequent{}{\update \upl v:= e\upr\GlobalJudge{s}{\Phi}}$ is valid if and only if
$\sequent{}{\GlobalJudge{\update  v=e;s}{\Phi}}$ is valid.

\paragraph{Rule $\mathsf{Cond}$}
We prove soundness and reversibility of the rule.
Let $\sigma$ be a state. We have, with
$T=\eval{\update}{\singleton{\sigma}}$:
$$\begin{array}{lll}
& & \sigma \models \judge{\update}{\noEvent{} \stateFml{e}}  \>\Rightarrow\>
\sigma \models \GlobalJudge{\update s; s'}{\Phi} 
\ \wedge\\
&& \sigma \models \judge{\update}{\noEvent{} \stateFml{!e}}   \>\Rightarrow\>
\sigma \models \GlobalJudge{\update s'}{\Phi} \\
\mbox{\footnotesize\{Def.~\ref{def:judgement-global-update}\}} 
&\Leftrightarrow~&
\sigma \models \judge{\update}{\noEvent{} \stateFml{e}}  \>\Rightarrow\>
T \chopSem \globEval{s; s'}{T} \subseteq \evalPhiNoArg{\Phi} 
\ \wedge\\
& &\sigma \models \judge{\update}{\noEvent{} \stateFml{!e}}  \>\Rightarrow\>
T \chopSem  \globEval{s'}{T} \subseteq \evalPhiNoArg{\Phi} \qquad  
\\
\mbox{\footnotesize\{Prop.~\ref{prop:glob-sem-comp-update}\}}&\Leftrightarrow~& 
\sigma \models \judge{\update}{\noEvent{} \stateFml{e}}  \>\Rightarrow\>
T \,\chopSem\,
\globEval{s;s'}{T}  \subseteq \evalPhiNoArg{\Phi} 
\wedge\\
&& \sigma \models \judge{\update}{\noEvent{} \stateFml{!e}}  \>\Rightarrow\>
T \,\chopSem\,
\globEval{s'}{T} \subseteq \evalPhiNoArg{\Phi} 
 \\
 \mbox{\footnotesize\{Def.~\ref{def:local-update-judgement}\}}&\Leftrightarrow~&
 \forall \tau\in T.\\
 &&(\valBP{last(\tau)}{}{e} = \trueSem  \>\Rightarrow\>
\tau \,\chopSem\,
\globEval{s;s'}{\tau}  \subseteq \evalPhiNoArg{\Phi} \wedge\\
&&\valBP{last(\tau)}{}{e} = \falseSem  \>\Rightarrow\>
\tau \,\chopSem\,
\globEval{s'}{\tau} \subseteq \evalPhiNoArg{\Phi})
\\
\mbox{\footnotesize\{Fig.~\ref{fig:local}, Def.~\ref{def:glob-sem}\}}
&\Leftrightarrow~&
T \,\chopSem\,
\globEval{\ifStmt{e}{s};s'}{T}
\subseteq \evalPhiNoArg{\Phi} 
 \\
 \mbox{\footnotesize\{Def.~\ref{def:judgement-global-update}\}} 
 &\Leftrightarrow~&
 \sigma \models \update\GlobalJudge{\ifStmt{e}{s}; s'}{\Phi} 
\\
\end{array} $$
\emph{Explanation:}
Given Fig.~\ref{fig:local}, Def.~\ref{def:glob-sem} we have
\begin{align}
  &\semSeq
    {\semPair{\tau}{\ifStmt{e}{s};s'}}
    {\semPair{\tau}{s;s'}}{*}, \  \text{if} \ \valB{last(\tau)}{e}=\trueSem, \ \text{and}\notag\\
  &\semSeq
    {\semPair{\tau}{\ifStmt{e}{s};s'}}
    {\semPair{\tau}{s'}}{*}, \  \text{if} \ \valB{last(\tau)}{e}=\falseSem \notag
\end{align}
Therefore 
\begin{align}
  &\globEval{\ifStmt{e}{s};s'}{\tau}=\globEval{s;s'}{\tau} \ \text{if} \ \valB{last(\tau)}{e}=\trueSem, \ \text{and}\notag\\
  &\globEval{\ifStmt{e}{s};s'}{\tau}=\globEval{s'}{\tau} \  \text{if} \ \valB{last(\tau)}{e}=\falseSem \notag
\end{align}

\paragraph{Rule $\mathsf{Return}$}
We prove soundness and reversibility of the rule.
Let $\sigma$ be a state. We have, with
$T=\eval{\update}{\singleton{\sigma}}$:

$$\begin{array}{lll}
  && \sigma \models \GlobalJudge{\update\upP{\retEv(oId)}}{\Phi}\\
  \mbox{\footnotesize\{Def.~\ref{def:judgement-global-update}\}} 
  &\Leftrightarrow& 
  \eval{\update\upP{\retEv(oId)}}{\singleton{\sigma}} \chopSem 
  \globEval{\zero}{\eval{\update\upP{\retEv(oId)}}{\singleton{\sigma}}}\subseteq\eval{\Phi}{}\\
  \mbox{\footnotesize\{Prop.~\ref{prop:glob-sem-comp-update}\}} 
  &\Leftrightarrow&
  T \chopSem \eval{\upP{\retEv(oId)}}{T} \chopSem \globEval{\zero}{T \chopSem \eval{\upP{\retEv(oId)}}{T}}\subseteq\eval{\Phi}{}\\
  \mbox{\footnotesize\{Fig.~\ref{fig:stmt-update-eval}\}} 
  &\Leftrightarrow&  
  T \chopSem \locEval{\xasync{return}}{T} \chopSem \globEval{\zero}{T \chopSem \locEval{\xasync{return}}{T}}\subseteq\eval{\Phi}{}\\
  \mbox{\footnotesize\{Prop.~\ref{prop:semantics-composition}\}} 
  &\Leftrightarrow&  
  T \chopSem \globEval{\xasync{return}}{T}\subseteq\eval{\Phi}{}\\
  \mbox{\footnotesize\{Def.~\ref{def:judgement-global-update}\}}
  &\Leftrightarrow&  
  \sigma \models \GlobalJudge{\update\xasync{return}}{\Phi}\\
\end{array} $$

We, therefore, have that the sequent
$\sequent{}{\update \upP{\retEv(oId)}\GlobalJudge{}{\Phi}}$ is valid if and only if
$\sequent{}{\GlobalJudge{\update  \xasync{return}}{\Phi}}$ is valid.

\paragraph{Rule $\mathsf{AsyncCall}$}
We prove soundness and reversibility of the rule.
Let $\sigma$ be a state. We have, with
$T=\eval{\update}{\singleton{\sigma}}$, and a fresh call identifier $i$:

$$\begin{array}{lll}
  && \sigma \models \GlobalJudge{\update\upP{\invocEvP{}{m,i}}s}{\Phi}\\
  \mbox{\footnotesize\{Def.~\ref{def:judgement-global-update}\}} 
  &\Leftrightarrow& 
  \eval{\update\upP{\invocEvP{}{m,i}}}{\singleton{\sigma}} \chopSem \globEval{s}{\eval{\update\upP{\invocEvP{}{m,i}}}{\singleton{\sigma}}}\subseteq\eval{\Phi}{}\\
  \mbox{\footnotesize\{Prop.~\ref{prop:glob-sem-comp-update}\}} 
  &\Leftrightarrow&
  T \chopSem \eval{\upP{\invocEvP{}{m,i}}}{T} \chopSem \globEval{s}{T \chopSem \eval{\upP{\invocEvP{}{m,i}}}{T}}\subseteq\eval{\Phi}{}\\
  \mbox{\footnotesize\{Fig.~\ref{fig:local},\ref{fig:stmt-update-eval}\}} 
  &\Leftrightarrow&  
  T \chopSem \locEval{\xasync{!m()}}{T} \chopSem \globEval{s}{T \chopSem \locEval{\xasync{!m()}}{T}}\subseteq\eval{\Phi}{}\\
  \mbox{\footnotesize\{Prop.~\ref{prop:semantics-composition}\}} 
  &\Leftrightarrow&  
  T \chopSem \globEval{\xasync{!m()};s}{T}\subseteq\eval{\Phi}{}\\
  \mbox{\footnotesize\{Def.~\ref{def:judgement-global-update}\}}
  &\Leftrightarrow&  
  \sigma \models \GlobalJudge{\update\xasync{!m()};s}{\Phi}\\
\end{array} $$

\paragraph{Rule $\mathsf{Open}$}
We prove soundness and reversibility of the rule.

Let $\sigma$ be a state. We have, with
$T=\eval{\update}{\singleton{\sigma}}$:

$$\begin{array}{lll}
  && \sigma \models \GlobalJudge{\update\upP{\openEvP{}{f}}s}{\Phi}\\
  \mbox{\footnotesize\{Def.~\ref{def:judgement-global-update}\}} 
  &\Leftrightarrow& 
  \eval{\update\upP{\openEvP{}{f}}}{\singleton{\sigma}} \chopSem \globEval{s}{\eval{\update\upP{\openEvP{}{f}}}{\singleton{\sigma}}}\subseteq\eval{\Phi}{}\\
  \mbox{\footnotesize\{Prop.~\ref{prop:glob-sem-comp-update}\}} 
  &\Leftrightarrow&
  T \chopSem \eval{\upP{\openEvP{}{f}}}{T} \chopSem \globEval{s}{T \chopSem \eval{\upP{\openEvP{}{f}}}{T}}\subseteq\eval{\Phi}{}\\
  \mbox{\footnotesize\{Fig.~\ref{fig:local},\ref{fig:stmt-update-eval}\}} 
  &\Leftrightarrow&  
  T \chopSem \locEval{\xasync{open(}f\xasync{)}}{T} \chopSem \globEval{s}{T \chopSem \locEval{\xasync{open(}f\xasync{)}}{T}}\subseteq\eval{\Phi}{}\\
  \mbox{\footnotesize\{Prop.~\ref{prop:semantics-composition}\}} 
  &\Leftrightarrow&  
  T \chopSem \globEval{\xasync{open(}f\xasync{)};s}{T}\subseteq\eval{\Phi}{}\\
  \mbox{\footnotesize\{Def.~\ref{def:judgement-global-update}\}}
  &\Leftrightarrow&  
  \sigma \models \GlobalJudge{\update\xasync{open(}f\xasync{)};s}{\Phi}\\
\end{array} $$

\paragraph{Rule $\mathsf{Close}$}
The proof is analogous to the proof for Rule \textsf{Open}, but the rule is not reversible due to the additional premise needed for Thm.~\ref{thm:global}.

\paragraph{Rule $\mathsf{Read}$}
The proof is analogous to the proof for Rule \textsf{Open}, but the rule is not reversible.

\paragraph{Rule $\mathsf{Write}$}
The proof is analogous to the proof for Rule \textsf{Open}, but the rule is not reversible.

\paragraph{Rule $\mathsf{Finish}$}
We prove soundness and reversibility of the rule.
Let $\sigma$ be a state. We have, with
$T=\eval{\update}{\singleton{\sigma}}$:

$$\begin{array}{lll}
  && \sigma \models \judge{\update\upP{\popEvP{}{(m,oId)}}}{\Phi} 
  \wedge schedule(\update)=\emptyset\\
  \mbox{\footnotesize\{Def.~\ref{def:judgement-global-update}\}} 
  &\Leftrightarrow& 
  \eval{\update\upP{\popEvP{}{(m,oId)}}}{\singleton{\sigma}} \subseteq\eval{\Phi}{}  
  \wedge schedule(\update)=\emptyset\\
  \mbox{\footnotesize\{Prop.~\ref{prop:glob-sem-comp-update}\}} 
  &\Leftrightarrow&
  T \chopSem \eval{\upP{\popEvP{}{(m,oId)}}}{T} \subseteq\eval{\Phi}{} 
   \wedge schedule(\update)=\emptyset\\
  \mbox{\footnotesize\{Lemma~\ref{lemma:soundness-scheduling}\}} 
  &\Leftrightarrow&
  T \chopSem \eval{\upP{\popEvP{}{(m,oId)}}}{T} \subseteq\eval{\Phi}{} 
   \wedge \bigcup_{\tau \in T} schedule(\tau)=\emptyset\\
  \mbox{\footnotesize\{Fig.~\ref{fig:stmt-update-eval}\}}  &\Leftrightarrow&
  \forall \tau \in T. \ (\tau \chopSem \popEvP{last(\tau)}{(m,oId)} \subseteq\eval{\Phi}{} 
  \wedge schedule(\tau)=\emptyset)\\
  \mbox{\footnotesize\{Fig.~\ref{fig:global}, Def.~\ref{def:glob-sem}\}}  &\Leftrightarrow&
  T  \chopSem \globEval{\zero}{T} \subseteq\eval{\Phi}{} \\
  \mbox{\footnotesize\{Def.~\ref{def:judgement-global-update}\}}  &\Leftrightarrow&
  \sigma \models \GlobalJudge{\update}{\Phi}\\
\end{array} $$

To make the following proof more readable we generalize Prop.~\ref{prop:semantics-composition}:
\begin{proposition}
  \label{prop:composition-global}
  $
  \globEval{s;s'}{\tau} = \globEval{s}{\tau} \chopSem \globEval{s'}{\globEval{s}{\tau}}
  $
\end{proposition}
\begin{proof}
This proposition follows from the fact that 
$$
  \globEval{\zero}{\tau} \chopSem \globEval{s}{\globEval{\zero}{\tau}}
  = \globEval{\zero; s}{\tau}
  = \globEval{s}{\tau}\enspace.
  $$
If $\globEval{\zero}{\tau} = \{\singleton{first(\tau)}\}$ it means that no 
asynchronous calls can be scheduled and the equivalence is trivially valid.
If $\globEval{\zero}{\tau} \neq \{\singleton{first(\tau)}\}$ 
it means that $\tau$ triggers the rule {\normalfont$\textsc{Run}$}, 
no matter what  is the content of the continuation.
\end{proof}

\paragraph{Rule $\mathsf{Call}$}
We prove the soundness of the rule.
Let $\sigma$ be a state. We have with
$T=\eval{\update}{\singleton{\sigma}}$, $\Phi_a = (\Phi \wedge \theta_{a_m})$, $\Psi_c = (\Psi \wedge \theta_{c_m})$, $i$ fresh,
and $\mathsf{lookup}(m,\mathcal{G}) = m()\, \{\xasync{s'; return}\}$:
$$\begin{array}{lll}
  && \sigma \models \judge{\update}{\Phi_a} \ \wedge \\
  && \eval{\theta_{s_m}}{} \subseteq \eval{\theta}{} \ \wedge\\
  && (\sigma \models \judge{\update \upP{\runEv(m,i,sy)}}{\Phi_a \chop \, \theta_{s_m}}  \\
  && \quad \Rightarrow \sigma \models \GlobalJudge{\update \upP{\runEv(m,i,sy)}s}{\Phi_a \chop \, \theta_{s_m}} \chop \, \Psi_c)\\
  \mbox{\footnotesize\{Def.~\ref{def:local-update-judgement},\ref{def:judgement-global-update}\}} 
  &\Leftrightarrow&
  T \subseteq \eval{\Phi_a}{} \ \wedge \\
  && \eval{\theta_{s_m}}{} \subseteq \eval{\theta}{} \ \wedge\\
  &&(\eval{\update\upP{\runEv(m,i,sy)}}{\singleton{\sigma}} \subseteq \eval{\Phi_a \chop \, \theta_{s_m}}{} \\
  && \ \Rightarrow \eval{\update\upP{\runEv(m,i,sy)}}{\singleton{\sigma}} \chopSem \globEval{s}{\eval{\update\upP{\runEv(m,i,sy)}}{\singleton{\sigma}}} \subseteq \eval{\Phi_a \chop \, \theta_{s_m} \chop \, \Psi_c}{})\\
  \mbox{\footnotesize\{Prop.~\ref{prop:glob-sem-comp-update}\}} 
  &\Leftrightarrow&
  T \subseteq \eval{\Phi_a}{} \ \wedge \\
  && \eval{\theta_{s_m}}{} \subseteq \eval{\theta}{} \ \wedge\\
  &&(T \chopSem \eval{\upP{\runEv(m,i,sy)}}{T} \subseteq \eval{\Phi_a \chop \, \theta_{s_m}}{}\\
  && \ \Rightarrow T \chopSem \eval{\upP{\runEv(m,i,sy)}}{T} \chopSem \globEval{s}{T \chopSem \eval{\upP{\runEv(m,i,sy)}}{T}} \subseteq \eval{\Phi_a \chop \, \theta_{s_m} \chop \, \Psi_c}{})\\
  \mbox{\footnotesize\{Fig.~\ref{fig:stmt-update-eval}\}} 
  &\Leftrightarrow&
  T \subseteq \eval{\Phi_a}{} \ \wedge \\
  && \eval{\theta_{s_m}}{} \subseteq \eval{\theta}{} \ \wedge\\
  &&\forall \tau \in T. (\tau \chopSem \callEvP{last(\tau)}{m,i} \chopSem \globEval{\xasync{s'; return}}{\tau \chopSem \callEvP{last(\tau)}{m,i}} \subseteq \eval{\Phi_a \chop \, \theta_{s_m}}{} \\
  && \ \Rightarrow \tau \chopSem \callEvP{last(\tau)}{m,i} \chopSem \globEval{\xasync{s'; return}}{\tau \chopSem \callEvP{last(\tau)}{m,i}} \\
  && \ \quad \chopSem  \globEval{s}{\tau \chopSem \callEvP{last(\tau)}{m,i} \chopSem \globEval{\xasync{s'; return}}{\tau \chopSem \callEvP{last(\tau)}{m,i}}} \subseteq \eval{\Phi_a \chop \, \theta_{s_m} \chop \, \Psi_c}{})\\
  \mbox{\footnotesize\{Def.~$\globEval{m()}{\tau}$\}} 
  &\Leftrightarrow&
  T \subseteq \eval{\Phi_a}{} \ \wedge \\
  && \eval{\theta_{s_m}}{} \subseteq \eval{\theta}{} \ \wedge\\
  &&(T \chopSem \globEval{\xasync{m()}}{T} \subseteq \eval{\Phi_a \chop \, \theta_{s_m}}{} \\
  && \ \Rightarrow T \chopSem  \globEval{\xasync{m()}}{T} \chopSem  \globEval{s}{T \chopSem  \globEval{\xasync{m()}}{T}} \subseteq \eval{\Phi_a \chop \, \theta_{s_m} \chop \, \Psi_c}{})\\
  \mbox{\footnotesize\{Prop.~\ref{prop:composition-global}\}} 
  &\Leftrightarrow&
  T \subseteq \eval{\Phi_a}{} \ \wedge \\
  && \eval{\theta_{s_m}}{} \subseteq \eval{\theta}{} \ \wedge\\
  &&(T \chopSem \globEval{\xasync{m()}}{T} \subseteq \eval{\Phi_a \chop \, \theta_{s_m}}{}  \Rightarrow T \chopSem  \globEval{\xasync{m()};s}{T}\subseteq \eval{\Phi_a \chop \, \theta_{s_m} \chop \, \Psi_c}{})\\
  \mbox{\footnotesize\{Fig.~\ref{fig:semantics-formulas}, weak.\}} 
  &\Rightarrow&
  T \subseteq \eval{\Phi}{} \ \wedge \\
  && \eval{\theta_{s_m}}{} \subseteq \eval{\theta}{} \ \wedge\\
  &&(T \chopSem \globEval{\xasync{m()}}{T} \subseteq \eval{\Phi}{} \chopSem \eval{\theta_{s_m}}{} \Rightarrow T \chopSem  \globEval{\xasync{m()};s}{T}\subseteq \eval{\Phi}{} \chopSem \eval{\theta_{s_m}}{} \chopSem \, \eval{\Psi}{})\\
  \mbox{\footnotesize\{$T \subseteq \eval{\Phi_a}{}$\}} 
  &\Rightarrow&
  \eval{\theta_{s_m}}{} \subseteq \eval{\theta}{} \ \wedge\\
  &&(\globEval{\xasync{m()}}{T} \subseteq \eval{\theta_{s_m}}{} \Rightarrow T \chopSem  \globEval{\xasync{m()};s}{T}\subseteq \eval{\Phi}{} \chopSem \eval{\theta_{s_m}}{} \chopSem \, \eval{\Psi}{})\\
  \mbox{\footnotesize\{Contract of $m$\}} 
  &\Rightarrow&
  \eval{\theta_{s_m}}{} \subseteq \eval{\theta}{} \ \wedge
  T \chopSem  \globEval{\xasync{m()};s}{T}\subseteq \eval{\Phi}{} \chopSem \eval{\theta_{s_m}}{} \chopSem \, \eval{\Psi}{}\\
  \mbox{\footnotesize\{$\eval{\theta_{s_m}}{} \subseteq \eval{\theta}{}$\}} 
  &\Rightarrow&
  T \chopSem  \globEval{\xasync{m()};s}{T}\subseteq \eval{\Phi \chop \, \theta \chop \, \Psi}{}\\
  \mbox{\footnotesize\{Def.~\ref{def:judgement-global-update}\}} 
  &\Rightarrow&
  \sigma \models \GlobalJudge{\update\xasync{m()};s}{\Phi \chop \, \theta \chop \, \Psi}
\end{array} $$

\paragraph{Rules $\mathsf{ScheduleD}/\mathsf{ScheduleN}$}
Given the similar premises and similar conclusion, the proof for \textsf{ScheduleD} is analogous to the proof for the $\mathsf{Call}$ rule,
with the only difference that the $\callEvP{}{m,i}$ event is not generated.
This also applies to \textsf{ScheduleN} since this rule generalizes \textsf{ScheduleD} to be applied when more than one invocation can be scheduled.

\subsection{Proof Theorem~\ref{thm:global}}

We prove the two properties separately.
In general, the proof follows the structure of trace contract adherence of BPL~\cite{DBLP:conf/tableaux/Kamburjan19}.
\begin{description}
\item[File Correctness]
Suppose there would be a trace that is not file-correct, and a first event within that trace such that the trace up to this point is not file correct.
As every event is added by a procedure, there must be a procedure $m$ that violates file correctness.
Let $s$ be the operation that adds this event within $m$.
By assumption, all procedure adhere to their procedure contracts, thus $m$ also does so as well. 
Thus, we have a proof of $m$, which must symbolically execute $s$\footnote{This relies on a property of the symbolic execution that no statement is skipped. The technical details to show this property are tedious to set up and bear no insights, we refer to Kamburjan~\cite{DBLP:phd/dnb/Kamburjan20} for a detailed treatment.}. 
Consider the case where the trace is not file correct because of a read event without a preceding open (and no closing) on the same file $f$.
This means that at the moment of the execution of the read statement, the formula $\noEvent{}\openEvP{}{f}\noEvent{\closeEvP{}{f}}$ would not hold.
However, this is checked explicitly by the $\mathsf{Read}$ rule. Thus, this rule cannot succeed, the proof obligation cannot be closed and the procedure is not adherent.
This contradicts our assumptions, so the trace must be file-correct.
The other cases are analogous.

\item[Adherence]
From strong procedure adherence, program correctness follows by definition -- we must show that each procedure strongly adheres to its contract, 
given that (a) each procedure weakly adheres to its contracts, and (b) we have a proof in our calculus for this.

The argument resembles the one for file correctness. Let $m$ be a procedure and assume there is a trace that violated procedure adherence, there must be a first event or state where
the property does not hold. Let $m'$ be the caller of $m$ for this call identifier.
Now, $m'$ does have a closed proof for its weak adherence, thus, the call to $m$, whether synchronous or asynchronous, was symbolically executed.
Both rules to do so, add the post-condition trace to the post-condition of the verification target. As the proof is closed, this condition holds.
Thus, all the procedure $m$ must be strongly adherent if it is weakly adherent and always called correctly, which follows from Theorem~\ref{thm:sound}.
\end{description}
As a final note, we point out that the premises we reference in this proof are not needed for soundness -- Theorem~\ref{thm:sound} holds if they are removed.

\section{Properties}
\begin{proposition}[Wellformedness of Call Tree]
  Let $\tau$ be a trace and $scp$ be a call scope.
  If $scp \in V_{idle}(\tau) \implies children(scp, \tau) = \emptyset$, i.e., $scp$ is a leaf.
\end{proposition}
  \begin{proof}
    Let $\tau$ be a trace and $(m,i) \in V_{idle}$.
  By Def.~\ref{def:call-tree}, 
  $\tau \in \finiteNoEv{} \invocEvP{}{m,i} \finiteNoEv{\callEvP{}{m,i}}$
  and therefore $\tau \in \finiteNoEv{\callEvP{}{m,i}}$, since an event $\callEvP{}{m,i}$
  can never occur before an event $\invocEvP{}{m,i}$.
  By Def.~\ref{def:lastev-currctx}, there is no $\tau'$ such that $\tau \in \tau' \finiteNoEv{}$ 
  and $currScp(\tau') = (m,i)$. 
  Therefore $children((m,i), \tau)=\emptyset$.
  \end{proof}

\section{Further Proofs}
This section contains the proof for the contract of \async{do}.
We use again the following abbreviations:
\begin{align*}
\update &= \mathcal{V}\{\startEv(\xasync{do},\mathit{oId})\}\\
\phi_2 &= \judge{\update}{\noEvent{\openEvP{}{f}}\stateFml{\mathsf{true}}} \\
\theta_1 &= \noEvent{\openEvP{}{f}}\stateFml{\mathsf{true}} \noEvent{}\closeEvP{}{f}\noEvent{} 
\end{align*}

We again apply the $\mathsf{Contract}$ and perform symbolic execution using $\mathsf{Open}$, $\mathsf{AsyncCall}$
before arriving at the synchronous call. Here, we proof the pre-condition using the accumulated update and continue below.

\begin{prooftree}
\AxiomC{}
\UnaryInfC{$C_\mathtt{do}^M, \phi_2\vdash \judge{\update\{\openEvP{}{f}\}\{\invocEv(\xasync{closeF},i)\}}{\openEvP{}{f}\noEvent{\closeEvP{}{f}}\stateFml{\mathsf{true}}} $}
\AxiomC{$\vdots(3)$}
\LeftLabel{$\mathsf{SyncCall}$}
\BinaryInfC{$C_\mathtt{do}^M, \phi_2\vdash \GlobalJudge{\update\{\openEvP{}{f}\}\{\invocEv(\xasync{closeF},i)\} \xasync{operate(); return;}}{\theta_1}$}
\UnaryInfC{$\vdots$}
\UnaryInfC{$C_\mathtt{do}^M, \phi_2\vdash \GlobalJudge{\update \xasync{open(file); !closeF(); operate(); return;}}{
    \theta_1
    }$}
\LeftLabel{$\mathsf{Contract}$}
\UnaryInfC{$ C_\mathtt{do}^M \vdash \GlobalJudge{\mathtt{do}}{C_\mathtt{do}} $}
\end{prooftree}

\noindent We summarize the new pre-condition and update with the following abbreviations
\begin{align*}
U_2 &= \update\{\openEvP{}{f}\}\{\invocEv(\xasync{closeF},i)\}\{\runEv(\xasync{operate},i,sy)\}\\
\phi_3 &= \judge{\update\{\openEvP{}{f}\}\{\invocEv(\xasync{closeF},i)\}\{\runEv(\xasync{operate},i,sy)\}}
{
\noEvent{\openEvP{}{f}}\stateFml{\mathsf{true}}\concat\writeEvP{}{f}\noEvent{\closeEvP{}{f}}
}
\end{align*}
Next, we apply the $\mathsf{Return}$ rule and finally arrive at the handling of the invocation event -- the $\mathit{schedule}$ function now returns 
a singleton set that must take care of.
The premises checking the contract ($\vdots_\mathit{contr}$) are again trivial to show and omitted here.
Afterwards, we can again apply $\mathsf{Finish}$ as no invocation is left and trivially the post-condition, as now the $\closeEvP{}{}$ event
is added by the contract of \async{closeF}. 
We introduce a last abbreviation:
\[
\phi_4 = \judge{U_2\{\retEv(\mathit{oId})\}\{\runEv(\xasync{closeF},i,as)\}}{(\dots) \chop \stateFml{\mathsf{true}}\concat\closeEvP{}{f}\noEvent{\openEvP{}{f}}}
\]

\noindent\resizebox{\textwidth}{!}{
\begin{minipage}{1.15\textwidth}
\begin{prooftree}
\AxiomC{$\vdots$}
\AxiomC{}
\UnaryInfC{$\vdots$}
\UnaryInfC{$
C_\mathtt{do}^M, \phi_2, \phi_3, \phi_4
\vdash \judge{
U_2\{\retEv(\mathit{oId})\}\{\runEv(\xasync{closeF},i,as)\}
}{
\theta_1 \wedge \noEvent{}\closeEvP{}{f}\noEvent{}
}
$}
\LeftLabel{$\mathsf{Finish}$}
\UnaryInfC{$
C_\mathtt{do}^M, \phi_2, \phi_3, \phi_4
\vdash \GlobalJudge{
U_2\{\retEv(\mathit{oId})\}\{\runEv(\xasync{closeF},i,as)\}
}{
\theta_1 \wedge \noEvent{}\closeEvP{}{f}\noEvent{}
}
$}
\LeftLabel{$\mathsf{ScheduleF}$}
\BinaryInfC{
$C_\mathtt{do}^M, \phi_2, \phi_3
\vdash \GlobalJudge{U_2\{\retEv(\mathit{oId})\}}{\theta_1 \wedge \noEvent{}\closeEvP{}{f}\noEvent{}}$
    }
\LeftLabel{$\mathsf{Return}$}
\UnaryInfC{
$C_\mathtt{do}^M, \phi_2, \phi_3
\vdash \GlobalJudge{U_2 \xasync{return;}}{\theta_1 \wedge \noEvent{}\closeEvP{}{f}\noEvent{}}$
    }
\UnaryInfC{$\vdots(3)$}
\end{prooftree}
\end{minipage}}


\end{document}
